\documentclass[letterpaper,USenglish11pt,]{article}
\usepackage[margin=1in]{geometry}
\usepackage{amsmath}
\usepackage{amssymb}
\usepackage{amsthm}
\usepackage{amsfonts}
\usepackage{dsfont}
\usepackage{mathrsfs}
\newtheorem{lemma}{Lemma}[section]
\newtheorem{theorem}[lemma]{Theorem}
\newtheorem{claim}[lemma]{Claim}
\newtheorem{definition}[lemma]{Definition}
\newtheorem{corollary}[lemma]{Corollary}
\usepackage{xcolor,xspace}
\usepackage{tikz}
\usepackage{hyperref}

\newcounter{note}[section]

\title{Lasserre Integrality Gaps for Graph Spanners and Related Problems}

\author{Michael Dinitz \qquad Yasamin Nazari \qquad Zeyu Zhang \\ Johns Hopkins University}

\date{}
\begin{document}

\allowdisplaybreaks

\maketitle

\begin{abstract}
There has been significant recent progress on algorithms for approximating graph spanners, i.e., algorithms which approximate the best spanner for a given input graph.  Essentially all of these algorithms use the same basic LP relaxation, so a variety of papers have studied the limitations of this approach and proved integrality gaps for this LP in a variety of settings.  We extend these results by showing that even the strongest lift-and-project methods cannot help significantly, by proving polynomial integrality gaps even for $n^{\Omega(\varepsilon)}$ levels of the Lasserre hierarchy, for both the directed and undirected spanner problems.  We also extend these integrality gaps to related problems, notably \textsc{Directed Steiner Network} and \textsc{Shallow-Light Steiner Network}.  
\end{abstract}

\section{Introduction}
A spanner is a subgraph which approximately preserves distances: formally, a $t$-spanner of a graph $G$ is a subgraph $H$ such that $d_H(u,v) \leq t \cdot d_G(u,v)$ for all $u,v \in V$ (where $d_H$ and $d_G$ denote shortest-path distances in $H$ and $G$ respectively).  Since $H$ is a subgraph it is also the case that $d_G(u,v) \leq d_H(u,v)$, and thus a $t$-spanner preserves all distances up to a multiplicative factor of $t$, which is known as the \emph{stretch}.  Graph spanners originally appeared in the context of distributed computing~\cite{PS89,PU89}, but have since been used as fundamental building blocks in applications ranging from routing in computer networks~\cite{TZ01} to property testing of functions~\cite{BGJRW09} to parallel algorithms~\cite{friedrichs18}.

Most work on graph spanners has focused on tradeoffs between various parameters, particularly the size (number of edges) and the stretch.  Most notably, a seminal result of Alth\"ofer et al.~\cite{ADDJS93} is that every graph admits a $(2k-1)$-spanner with at most $n^{1+1/k}$ edges, for every integer $k \geq 1$.  This tradeoff is also known to be tight, assuming the Erd\H{o}s girth conjecture~\cite{Erd64}, but extensions to this fundamental result have resulted in an enormous literature on graph spanners.

Alongside this work on tradeoffs, there has been a line of work on \emph{optimizing} spanners.  In this line of work, we are usually given a graph $G$ and a value $t$,  and are asked to find the $t$-spanner of $G$ with the fewest number of edges.  If $G$ is undirected then this is known as \textsc{Basic $t$-Spanner}, while if $G$ is directed then this is known as \textsc{Directed $t$-Spanner}.  The best known approximation for \textsc{Directed $t$-Spanner} is an $O(n^{1/2})$-approximation~\cite{BBMRY13}, while for \textsc{Basic $t$-Spanner} the best known approximations are $O(n^{1/3})$ when $t = 3$~\cite{BBMRY13} and when $t=4$~\cite{DZ16}, and $O(n^{\frac{1}{\lfloor (t+1)/2\rfloor}})$ when $t > 4$.  Note that this approximation for $t > 4$ is directly from the result of~\cite{ADDJS93} by using the trivial fact that the optimal solution is always at least $n-1$ (in a connected graph), and so is in a sense ``generic" as both the upper bound and the lower bound are universal, rather than applying to the particular input graph.  

One feature of the algorithms of~\cite{BBMRY13,DZ16}, as well as earlier work~\cite{DK11} and extensions to related settings (such as approximating fault-tolerant spanners~\cite{DK11-FT,DZ16} and minimizing the maximum degree~\cite{CD16}), is that they all use some variant of the same basic LP: a flow-based relaxation originally introduced for spanners by~\cite{DK11}.  The result of~\cite{BBMRY13} uses a slightly different LP (based on cuts rather than flows), but it is easy to show that the LP of~\cite{BBMRY13} is no stronger than the LP of~\cite{DK11}.  

The fact that for \textsc{Basic $t$-Spanner} we cannot do better than the ``generic" bound when $t > 4$, as well as the common use of a standard LP relaxation, naturally gives rise to a few questions.  Is it possible to do better than the generic bound when $t > 4$?  Can this be achieved with the basic LP?  Can we analyze this LP better to get improvements for \textsc{Directed $t$-Spanner}?  In other words: what is the power of convex relaxations for spanner problems?  It seems particularly promising to use lift-and-project methods to try for stronger LP relaxations, since one of the very few spanner approximations that uses a different LP relaxation was the use of the Sherali-Adams hierarchy to give an approximation algorithm for the \textsc{Lowest Degree $2$-Spanner} problem~\cite{CDK12}.

It has been known since~\cite{EP07,Kor01} that \textsc{Directed $t$-Spanner} does not admit an approximation better than $2^{\log^{1-\varepsilon} n}$ for any constant $\varepsilon > 0$, and it was more recently shown in~\cite{DKR16} that \textsc{Basic $t$-Spanner} cannot be approximated any better than $2^{(\log^{1-\varepsilon} n) / t}$ for any constant $\varepsilon > 0$.  Thus no convex relaxation, and in particular the basic LP, can do better than these bounds.  But it is possible to prove stronger integrality gaps: it was shown in~\cite{DK11} that the integrality gap of the basic LP for \textsc{Directed $t$-Spanner} is at least $\tilde{\Omega}(n^{\frac{1}{3}-\varepsilon})$, while in~\cite{DZ16} it was shown that the basic LP for \textsc{Basic $t$-Spanner} has an integrality gap of at least $\Omega(n^{\frac{2}{(1+\varepsilon)(t+1) + 4}})$, which nearly matches the generic upper bound (particularly for large $t$).  

But this left open a tantalizing prospect: perhaps there are stronger relaxations which could be used to get improved approximation bounds.  Of course, the hardness of approximation results prove a limit to this.  But even with the known hardness results and integrality gaps, it is possible that there is, say, an $O(n^{1/1000})$-approximation for \textsc{Directed $t$-Spanner} and an $O(n^{1/(1000t)})$-approximation for \textsc{Basic $t$-Spanner} that uses more advanced relaxations.

\subsection{Our Results and Techniques}
This is the problem which we investigate: can we design stronger relaxations for spanners and related problems?  While we cannot rule out all possible relaxations, we show that an extremely powerful lift-and-project technique, the Lasserre hierarchy \cite{lasserre2001}, does not give relaxations which are massively better than the basic LP.  This is true despite the fact that Lasserre is an SDP hierarchy rather than an LP hierarchy, and despite the fact that we allow a polynomial number of levels in the hierarchy even though it can only be efficiently solved for a constant number of levels.  And since the Lasserre hierarchy is at least as strong as other hierarchies such as the Sherali-Adams hierarchy~\cite{sherali1990} and the Lovasz-Schrijver hierarchy~\cite{lovasz1991}, our results also imply integrality gaps for these hierarchies.  

Slightly more formally, we first rewrite the basic LP in a way that is similar to~\cite{BBMRY13} but is equivalent to the stronger original formulation~\cite{DK11}.  This makes the Lasserre lifts of the LP easier to reason about, thanks to the new structure of this formulation. We then consider the Lasserre hierarchy applied to to this LP, and prove the following theorems.

\begin{theorem}\label{thm:directSpanner}
For every constant $0<\varepsilon<1$ and sufficiently large $n$, the integrality gap of the $n^{\Omega(\varepsilon)}$-th level Lasserre SDP for \textsc{Directed $(2k-1)$-Spanner} is at least $\left(\frac{n}{k}\right)^{\frac{1}{18}-\Theta(\varepsilon)}$.
\end{theorem}

\begin{theorem}\label{thm:undirectSpanner}
For every constant $0<\varepsilon<1$ and sufficiently large $n$, the integrality gap of the $n^{\Omega(\varepsilon)}$-th level Lasserre SDP for \textsc{Basic $(2k-1)$-Spanner} is at least $\frac{1}{k}\cdot\left(\frac{n}{k}\right)^{\min\left\{\frac{1}{18},\frac{5}{32k-6}\right\}-\Theta(\varepsilon)}=n^{\Theta(\frac{1}{k}-\varepsilon)}$.
\end{theorem}

Note that, while the constant in the exponent is different, Theorem~\ref{thm:undirectSpanner} is similar to~\cite{DZ16} in that it shows that the integrality gap ``tracks" the trivial approximation from~\cite{ADDJS93} as a function of $k$.  Thus for undirected spanners, even using the Lasserre hierarchy cannot give too substantial an improvement over the trivial greedy algorithm.

At a very high level, we follow the approach to building spanner integrality gaps of~\cite{DK11,DZ16}.  They started with random instances of the \textsc{Unique Games} problem, which could be shown probabilistically to not admit any good solutions.  They then used these \textsc{Unique Game} instances to build spanner instances with the property that every spanner had to be large (or else the \textsc{Unique Games} instance would have had a good solution), but by ``splitting flow" the LP could be very small.  
%
%

In order to apply this framework to the Lasserre hierarchy, we need to make a number of changes.  First, since \textsc{Unique Games} can be solved reasonably well by Lasserre~\cite{BRS11,CMM06}, starting with a random instance of \textsc{Unique Games} will not work.  Instead, we start with a more complicated problem known as \textsc{Projection Games} (the special case of \textsc{Label Cover} in which all the edge relations are functions).  An integrality gap for the Lasserre hierarchy for \textsc{Projection Games} was recently given by~\cite{CMMV17,manurangsi2015} (based on an integrality gap for CSPs from~\cite{tulsiani2009}), so we can use this as our starting point and try to plug it into the integrality gap framework of~\cite{DK11,DZ16} to get an instance of either directed or undirected spanners.  Unfortunately, the parameters and structure that we get from this are different enough from the parameters used in the integrality gap of the basic LP that we cannot use~\cite{DK11, DZ16} as a black box.  We need to reanalyze the instance using different techniques, even for the ``easy" direction of showing that there are no good integral solutions.  In order to do this, we also need some additional properties of the gap instance for \textsc{Projection Games} from~\cite{CMMV17} which were not stated in their original analysis.  So we cannot even use~\cite{CMMV17} as a black box.

The main technical difficulty, though, is verifying that there is a ``low-cost" fractional solution to the SDP that we get out of this reduction.  For the basic LP this is straightforward, but for Lasserre we need to show that the associated slack moment matrices are all PSD.  This turns out to be surprisingly tricky, but by decomposing these matrices carefully we can show that each matrix in the decomposition is PSD, and thus the slack moment matrices are PSD. At a high level, we decompose the slack moment matrices as a summation of several matrices in a way that allows us to use the consistency properties of the feasible solution to the \textsc{Projection Games} instance in~\cite{CMMV17} to show that the overall sum is PSD.  

Doing this requires us to use some nice properties of the feasible fractional solution provided by~\cite{CMMV17}, some of which we need to prove as they were not relevant in the original setting.  In particular, one important property which makes our task much easier is that their fractional solution actually satisfies \emph{all} of the edges in the \textsc{Projection Games} instance.  That is, their integrality gap is in a particular ``place": the fractional solution has value $1$ while every integral solution has much smaller value.  Because spanners and the other network design problems we consider are minimization problems (where we need to satisfy all demands in a cheap way), this is enormously useful, as it essentially allows us to use ``the same" fractional solution (as it will also be feasible for the minimization version since it satisfies all edges).  Technically, we end up combining this fact about the fractional solution of~\cite{CMMV17} with several properties of the Lasserre hierarchy to infer some more refined structural properties of the derived fractional solution for spanners, allowing us to argue that they are feasible for the Lasserre lifts.  


\paragraph{Extensions.} 
A number of other network design problems exhibit behavior that is similar to spanners, and we can extend our integrality gaps to these problems.  In particular, we give a new integrality gap for Lasserre for \textsc{Directed Steiner Network} (\textsc{DSN}) (also called \textsc{Directed Steiner Forest}) and \textsc{Shallow-Light Steiner Network} (\textsc{SLSN}) \cite{BDZ18}.  In \textsc{DSN} we are given a directed graph $G = (V, E)$ (possibly with weights) and a collection of pairs $\{(s_i, t_i)\}_{i \in [p]}$, and are asked to find the cheapest subgraph such that there is a $s_i$ to $t_i$ path for all $i \in [p]$.  In \textsc{SLSN} the graph is undirected, but each $s_i$ and $t_i$ is required to be connected within a global distance bound $L$. The best known approximation for \textsc{DSN} is an $O(n^{3/5 + \varepsilon})$-approximation for arbitrarily small constant $\varepsilon > 0$~\cite{CDKL17}, which uses a standard flow-based LP relaxation. We can use the ideas we developed for spanners to also give integrality gaps for the Lasserre lifts of these problems. We provide the theorems here; details and proofs can be found in Appendix~\ref{app:extension}.

\begin{theorem}\label{thm:DSN}
For every constant $0<\varepsilon<1$ and sufficiently large $n$, the integrality gap of the $n^{\Omega(\varepsilon)}$-th level Lasserre SDP for \textsc{Directed Steiner Network} is at least $n^{\frac{1}{16}-\Theta(\varepsilon)}$.
\end{theorem}

\begin{theorem}\label{thm:SLSN}
For every constant $0<\varepsilon<1$ and sufficiently large $n$, the integrality gap of the $n^{\Omega(\varepsilon)}$-th level Lasserre SDP for \textsc{Shallow-Light Steiner Network} is at least $n^{\frac{1}{16}-\Theta(\varepsilon)}$.
\end{theorem}


\paragraph{Lift-and-Project for Network Design.} Lift and project methods such as Sherali-Adams~\cite{sherali1990} and Lasserre~\cite{lasserre2001} have been studied and used extensively for approximation algorithms.  For example, strong results are known about their performance on CSPs \cite{tulsiani2009, schoenebeck2008}, independent set in hypergraphs \cite{chlamtac2007}, graph coloring \cite{chlamtac2008}, and Densest $k$-Subgraph \cite{bhaskara2012,CDK12}. However, there is surprisingly little known about the power of these hierarchies for network design problems (the main exception being Directed Steiner Tree~\cite{friggstad2014, rothvoss2011}).  We begin to address this gap by providing Lasserre integrality gaps for a variety of difficult network design problems (\textsc{Basic $t$-Spanner}, \textsc{Directed $t$-Spanner}, \textsc{Directed Steiner Network}, and \textsc{Shallow Light Steiner Network}).  Our results can be seen as general framework for proving Lasserre integrality gaps for these types of hard network design problems.

\section{Preliminaries: Lasserre Hierarchy}

The Lasserre hierarchy is a way of lifting a polytope to a higher dimensional space, and then optionally projecting this left back to the original space in order to get tighter relaxations.  The standard characterization for Lasserre is as follows \cite{lasserre2001, laurent2003, rothvoss2013}:

\begin{definition}[Lasserre Hierarchy] \label{def:lasserre} 
Let $A\in\mathbb{R}^{m\times n}$ and $b \in\mathbb{R}^m$, and define the polytope $K=\{x\in\mathbb{R}^n : Ax\ge b\}$. The $r$-th level of the Lasserre hierarchy $L_r(K)$ consists of the set of vectors $y\in [0,1]^{\mathscr{P}([n])}$ where $\mathscr{P}$ means the power set, and they satisfy the following constraints:
$$
y_\varnothing=1, \quad
M_{r+1}(y):=(y_{I\cup J})_{|I|,|J| \le r+1}\succeq 0, \quad \forall\ell\in[m]:
M^\ell_r (y):=\left( \sum_{i=1}^n A_{\ell i}y_{I\cup J\cup \{i\}}-b_\ell y_{I\cup J}\right)_{|I|, |J| \le r}\succeq 0.
$$
The matrix $M_{r+1}$ is called the \textit{moment matrix}, and the matrices $M^\ell_r$ are called the \textit{slack moment matrices}.
\end{definition}


Let us review (see, e.g., \cite{rothvoss2013}) multiple helpful properties that we will use later. We include proofs in Appendix~\ref{app:lasserre} for completeness. 
\begin{claim}\label{claim:=1}
If $M_r(y)\succcurlyeq 0$, $|I|\le r$, and $y_I=1$, then $y_{I\cup J}=y_J$ for all $|J| \le r$.
\end{claim}

\begin{claim}\label{claim:=0}
If $M_r(y) \succcurlyeq0$, $|I|\le r$, and $y_I=0$, then $y_{I\cup J}=0$ for all $|J|\le r$.
\end{claim}
 
\begin{lemma} \label{lem:simple_slack}
If $M_{r+1}(y)\succeq 0$ then for any $i\in [n]$ we have $M^{i,1}(y)=(y_{I\cup J\cup \{i \}})_{|I|,|J| \le r}\succeq 0$ and $M^{i,0}(y)=(y_{I\cup J}-y_{I\cup J\cup \{i \}})_{|I|,|J| \le r}\succeq 0$. 
\end{lemma}

\section{\textsc{Projection Games}: Background and Previous Work} \label{sec:projection}

In this section we discuss the \emph{\textsc{Projection Games}} problem, its Lasserre relaxation, and the integrality gap that was recently developed for it~\cite{CMMV17} which form the basis of our integrality gaps for spanners and related problems.
We begin with the problem definition.

\begin{definition}[\textsc{Projection Games}]
Given a bipartite graph $(L,R,E,\Sigma, \{\pi_e\}_{e\in E})$, where $\Sigma$ is the (label) alphabet set and $\pi_e:\Sigma\rightarrow\Sigma$ for each $e\in E$, the objective is to find a label assignment $\alpha:L\cup R\rightarrow\Sigma$ that maximizes $\sum_{e=(v_L,v_R)\in E}\mathds{1}_{\pi_e(\alpha(v_L))=\alpha(v_R)}$ (i.e. the number of edges $e=(v_L,v_R)$ where $\pi_e(\alpha(v_L))=\alpha(v_R)$, which we refer to as satisfied edges).
\end{definition}

We will sometimes use relation notation for the functions $\pi_e$, w.g., we will talk about $(\sigma_1, \sigma_2) \in \pi_e$.  Note that \textsc{Projection Games} is the famous \textsc{Label Cover} problem but where the relation for every edge is required to be a function (and hence we inherit the relation notation when useful).  Similarly if we further restrict every function $\pi_e$ to be a bijection then we have the \textsc{Unique Games} problem.  So \textsc{Projection Games} lies ``between" \textsc{Unique Games} and \textsc{Label Cover}.

The basis of our integrality gaps is the integrality gap instance recently shown by~\cite{CMMV17} for Lasserre relaxations of \textsc{Projection Games}.  We first formally define this SDP.  For every $\Psi \subseteq (L\cup R)\times\Sigma$ we will have a variable $y_{\Psi}$.  Then the $r$-th level Lasserre SDP for \textsc{Projection Games} is the following.

\[\begin{array}{rll}
\text{SDP}_{Proj}^r:~\max&\sum\limits_{(v_L,v_R)\in E,(\sigma_L,\sigma_r)\in\pi_{(u,v)}}y_{(v_L,\sigma_L),(v_R,\sigma_R)}&\\
s.t.&y_\varnothing=1&\\
&M_r(y)=\left(y_{\Psi_1\cup\Psi_2}\right)_{|\Psi_1|,|\Psi_2|\le r}\succcurlyeq0&\\
&M_r^v(y)=\left(\sum_{\sigma\in\Sigma}y_{\Psi_1\cup\Psi_2\cup\{(v,\sigma)\}}-y_{\Psi_1\cup\Psi_2}\right)_{|\Psi_1|,|\Psi_2|\le r}=\mathbf{0}&\forall v\in V\
\end{array}\]

It is worth noting that this is not the original presentation of this SDP given by~\cite{CMMV17}: they wrote it using a vector inner product representation.  But it can be shown that these representations are equivalent, and in particular we prove the important direction of this in Appendix~\ref{app:projEquiv}: any feasible solution to their version gives an equivalent feasible solution to $\text{SDP}_{Proj}^r$, and thus their fractional solutions are also fractional solutions to $\text{SDP}_{Proj}^r$.

\cite{CMMV17} gives a \textsc{Projection Games} instance with following properties. One of the properties is not proven in their paper, but is essentially trivial. We give a proof of this property, as well as a discussion of how the other properties follow from their construction, in Appendix \ref{app:projInstance}.

\begin{lemma}\label{lem:projInstance}
For any constant $0<\varepsilon<1$, there exists a \textsc{Projection Games} instance $(L,R,E_{Proj},\Sigma,\\(\pi_e)_{e\in E_{Proj}})$ with the following properties:
\begin{enumerate}
\item $\Sigma=[n^\frac{3-3\varepsilon}{5}]$, $R=\{x_1,\mathellipsis,x_n\}$, $L=\{c_1,\mathellipsis,c_m\}$, where $m=n^{1+\varepsilon}$.
\item There exists a feasible solution $\mathbf{y}^*$ for the $r=n^{\Omega(\varepsilon)}$-th level $\text{SDP}_{Proj}^r$, such that for all $\{c_i,x_j\}\in E_{Proj}$, we have ${\sum\limits_{(\sigma_L,\sigma_R)\in\pi_{(c_i,x_j)}}y_{(c_i,\sigma_L),(x_j,\sigma_R)}=1}$. 
\item At most $O\left(\frac{n^{1+\varepsilon}\ln n}{\varepsilon}\right)$ edges can be satisfied.
\item The degree of vertices in $L$ is $K=n^\frac{1-\varepsilon}{5}-1$, and the degree of vertices in $R$ is at most $2Kn^\varepsilon$.
\end{enumerate}
\end{lemma}

We also define $\pi_{i,j}=\pi_e$ if $e=\{c_i,x_j\}\in E_{Proj}$.

\section{Lasserre Integrality Gap for \textsc{Directed $(2k-1)$-Spanner}} \label{sec:directed}

%


In this section we prove our main result for the \textsc{Directed $(2k-1)$-Spanner} problem: a polynomial integrality gap for  polynomial levels of the Lasserre hierarchy.  We begin by discussing the base LP that we will use and its Lasserre lifts, then define the instance of \textsc{Directed $(2k-1)$-Spanner} that we will analyze (based on the integrality gap instance for \textsc{Projection Games} in Lemma \ref{lem:projInstance}), and then analyze this instance.

\subsection{Spanner LPs and their Lasserre lifts} \label{sec:spanner-lp}
%
%
%
The standard flow-based LP for spanners (including both the directed and basic $k$-spanner problems) was introduced by~\cite{DK11}, and has subsequently been used in many other spanner problems~\cite{BBMRY13,DK11-FT,CD16}.  Let $\mathcal P_{u,v}$ denote the set of all stretch-$k$ paths from $u$ to $v$.
\[\begin{array}{rll}
\text{LP}_{Spanner}^{Flow}:~\min&\sum\limits_{e\in E}x_e&\\
s.t.&\sum\limits_{P\in \mathcal{P}_{u,v}:e\in P}f_P\le x_e&\forall(u,v)\in E,\forall e\in E\\
&\sum\limits_{P\in \mathcal{P}_{u,v}}f_P\ge1&\forall(u,v)\in E\\
&x_e\ge0&\forall e\in E\\
&f_P\ge0&\forall(u,v)\in E, P\in \mathcal{P}_{u,v}
\end{array}\]

While this LP is extremely large (the number of variables can be exponential if there are general lengths on the edges, or if all lengths are unit but $k$ is large enough), it was shown in~\cite{DK11} that it can be solved in polynomial time.  However, for the purposes of studying its behavior in the Lasserre hierarchy, $\text{LP}_{Spanner}^{Flow}$ is a bit awkward.  Since it has (potentially) exponential size, so do its Lasserre lifts.  And from a more ``intuitive" point of view, since there are two different ``types" of variables, the lifts become somewhat difficult to reason about.  

Since the $f_P$ variables do not appear in the objective function, we can project the polytope defined by $\text{LP}_{Spanner}^{Flow}$ onto the $x_e$ variables and use the same objective function to get an equivalent LP but with only the $x_e$ variables.  More formally, let $\mathcal{Z}^{u,v}=\{\mathbf{z}\in[0,1]^{|E|} : \sum_{e\in P}z_e\ge1 \ \forall P \in \mathcal P_{u,v}\}$ be the polytope bounded by $0\le z_e\le1$ for all $e\in E$ and $\sum_{e\in P}z_e\ge1$ for all $P\in\mathcal{P}_{u,v}$.  Then it is not hard to see that if we project the polytope defined by $\text{LP}_{Spanner}^{Flow}$ onto just the $x_e$ variables, we get precisely the following LP (this can also be seen via the duality between stretch-$k$ flows and fractional cuts against stretch-$k$ paths):
\[\begin{array}{rll}
\text{LP}_{Spanner}:~\min&\sum\limits_{e\in E}x_e&\\
s.t.&\sum\limits_{e\in E}z_ex_e\ge1&\forall (u,v)\in E,\mbox{ and }\mathbf{z}\in\mathcal{Z}^{u,v}\\
&x_e\ge0&\forall e\in E\\
\end{array}\]

While as written there are an infinite number of constraints, it is easy to see by convexity that we need to only include the (exponentially many) constraints corresponding to vectors $\mathbf{z}$ that are vertices in the polytope $\mathcal Z^{u,v}$, for each $(u,v) \in E$.  Thus there are only an exponential number of constraints, but for simplicity we will analyze this LP as if there were constraints for all possible $\mathbf{z}$. This LP is completely equivalent to $\text{LP}_{Spanner}^{Flow}$, in the sense that a vector $\mathbf{x}$ is feasible for $\text{LP}_{Spanner}$ if any only if there exist $f_P$ variables so that $(\mathbf{x},\mathbf{f})$ is feasible for $\text{LP}_{Spanner}^{Flow}$.  The proof of the following theorem is included in Appendix~\ref{app:spanner-lp}

\begin{theorem} \label{thm:spanner-lp-equivalent}
$\mathbf{x}$ is feasible for $\text{LP}_{Spanner}$ if and only if there is some $f_P$ for each $(u,v) \in E$ and $P \in \mathcal P_{u,v}$ so that $(\mathbf{x}, \mathbf{f})$ is feasible for $\text{LP}_{Spanner}^{Flow}$.
\end{theorem}

From Definition~\ref{def:lasserre},  the $r$-th level Lasserre SDP of $\text{LP}_{Spanner}$ is:
\[\begin{array}{rll}
\text{SDP}_{Spanner}^r:~\min&\sum\limits_{e\in E}y_e&\\
s.t.&y_\varnothing=1&\\
&M_{r+1}(y)=\left(y_{I\cup J}\right)_{|I|,|J|\le r+1}\succcurlyeq0&\\
&M_r^\mathbf{z}(y)=\left(\sum_{e\in E}z_ey_{I\cup J\cup\{e\}}-y_{I\cup J}\right)_{|I|,|J|\le r}\succcurlyeq0&\forall (u,v)\in E,\mbox{ and }\mathbf{z}\in\mathcal{Z}^{u,v}\\
\end{array}\]

This SDP is the basic object of study in this paper, and is what we will prove integrality gaps about.

\subsection{Spanner Instance} \label{sec:dir_instance}

In this section we formally define the instance of \textsc{Directed $(2k-1)$-Spanner} that we will analyze to prove the integrality gap.  We basically follow the framework of~\cite{DK11}, who showed how to use the hardness framework of~\cite{EP07,Kor01} to prove integrality gaps for the basic flow LP.  We start with a different instance (integrality gaps instances for \textsc{Projection Games} rather than random instances of \textsc{Unique Games}), and also slightly change the reduction in order to obtain a better dependency on $k$.

Roughly speaking, given a \textsc{Projection Games} instance, we start with the ``label-extended" graph.  For each original vertex in the projection game, we create a group of vertices in the spanner instance of size $|\Sigma|$. So each vertex in the group can be thought as a label assignment for the \textsc{Projection Games} vertex.  We then add paths between these groups corresponding to each function $\pi_e$ (we add a path if the associated assignment satisfies the \textsc{Projection Games} edges).  We add many copies of the \textsc{Projection Games} graph itself as the ``outer edges'', and then connect each \textsc{Projection Games} vertex to the group associated with it. The key point is to prove that any integral solution must contain either many outer edges or many ``connection edges" (in order to span the outer edges), while the fractional solution can buy connection and inner edges fractionally and simultaneously span all of the outer edges.

More formally, given the \textsc{Projection Games} instance $(L=\cup_{i\in m}\{c_i\},R=\cup_{i\in n}\{x_i\},E_{Proj},\Sigma=[n^\frac{3-3\varepsilon}{5}],(\pi_e)_{e\in E_{Proj}})$ from Lemma \ref{lem:projInstance}, we create a directed $(2k-1)$-spanner instance $G=(V,E)$ as follows (note that $K$ is the degree of the vertices in $L$):

For every $c_i\in L$, we create $|\Sigma|+kK|\Sigma|$ vertices: $c_{i,\sigma}$ for all $\sigma\in\Sigma$ and $c_i^l$ for all $l\in[kK|\Sigma|]$. We also create edges $(c_i^l,c_{i,\sigma})$ for each $\sigma\in\Sigma$ and $l\in[kK|\Sigma|]$. We call this edge set $E_L$.

For every $x_i\in R$, we create $|\Sigma|+kK|\Sigma|$ vertices: $x_{i,\sigma}$ for $\sigma\in\Sigma$ and $x_i^l$ for $l\in[kK|\Sigma|]$. We also create edge $(x_i^l,x_{i,\sigma})$ for each $\sigma\in\Sigma$ and $l\in[kK|\Sigma|]$. We call this edge set $E_R$.

For every $e=\{c_i,x_j\}\in E_{Proj}$, we create edges $(c_i^l,x_j^l)$ for each $l\in[kK|\Sigma|]$. We call this edge set $E_{Outer}$.

For each $e=\{c_i,x_j\}\in E_{Proj}$ and $(\sigma_L,\sigma_R)\in\pi_{i,j}$, we also create vertices $w_{i,j,\sigma_L,\sigma_R,t}$ for $t\in[2k-4]$ and edges $(c_{i,\sigma_L},w_{i,j,\sigma_L,\sigma_R,1}),(w_{i,j,\sigma_L,\sigma_R,1},w_{i,j,\sigma_L,\sigma_R,2}),\mathellipsis,(w_{i,j,\sigma_L,\sigma_R,2k-4},x_{j,\sigma_R})$. We call this edge set $E_M$.

Finally, for technical reasons we needs some other edges $E_{LStars}$ and $E_{RStars}$ inside groups of $L_{Labels}$ and $R_{Labels}$, which will be defined later.

To be more specific, $V=L_{Dups}\cup L_{Labels}\cup M_{Paths}\cup R_{Labels}\cup R_{Dups}$, $E=E_L\cup E_{LStars}\cup E_M\cup E_{RStars}\cup E_R\cup E_{Outer}$, such that:

\[\begin{array}{lll}
L_{Labels}=\{c_{i,\sigma}\mid i\in|L|,\sigma\in\Sigma\},&\hspace{-6em}L_{Dups}=\{c_i^l\mid i\in|L|,l\in[kK|\Sigma|]\},\\
R_{Labels}=\{x_{i,\sigma}\mid i\in|R|,\sigma\in\Sigma\},&\hspace{-6em}R_{Dups}=\{x_i^l\mid i\in|R|,l\in[kK|\Sigma|]\},\\
M_{Paths}=\{w_{i,j,\sigma_L,\sigma_R,t}\mid\{c_i,x_j\}\in E_{Proj},&\hspace{-7.666em}(\sigma_L,\sigma_R)\in\pi_{i,j},t\in[2k-4]\}\\
E_L^{i,l}=\{(c_i^l,c_{i,\sigma})\mid\sigma\in\Sigma\},&\hspace{-6em}E_L^l=\cup_{i\in|L|}E_L^{i,l},&\hspace{-9em}E_L=\cup_{l\in[kK|\Sigma|]}E_L^l,\\
E_R^{i,l}=\{(x_i^l,x_{i,\sigma})\mid\sigma\in\Sigma\},&\hspace{-6em}E_R^l=\cup_{i\in|R|}E_R^{i,l},&\hspace{-9em}E_R=\cup_{l\in[kK|\Sigma|]}E_R^l,\\
E_M^{i,j,\sigma_L,\sigma_R}=\{(c_{i,\sigma_L},w_{i,j,\sigma_L,\sigma_R,1}),(w_{i,j,\sigma_L,\sigma_R,1},&\hspace{-5.333em}w_{i,j,\sigma_L,\sigma_R,2}),\mathellipsis,(w_{i,j,\sigma_L,\sigma_R,2k-4},x_{j,\sigma_R})\}\\
E_M^{i,j}=\cup_{(\sigma_L,\sigma_R)\in\pi_{i,j}}E_M^{i,j,\sigma_L,\sigma_R}&\hspace{-6em}E_M=\cup_{i,j:\{c_i,x_j\}\in E_{Proj}}E_M^{i,j}\\
E_{Outer}=\{(c_i^l,x_j^l)\mid\{c_i,x_j\}\in E_{Proj},l\in[kK|\Sigma|]\}\\
E_{LStars}=\{(c_{i,1},c_{i,\sigma}),(c_{i,\sigma},c_{i,1})\mid i\in|L|,\sigma\in\Sigma\setminus\{1\}\},\\
E_{RStars}=\{(x_{i,1},x_{i,\sigma}),(x_{i,\sigma},x_{i,1})\mid i\in|R|,\sigma\in\Sigma\setminus\{1\}\},
\end{array}\]

Note that if $k<3$, then there is no vertex set $M_{Paths}$, but only edge set $E_M$, which directly connect $c_{i,\sigma_L}$ and $x_{j,\sigma_R}$ for each $\{c_i,x_j\}\in E_{Proj}$ and $(\sigma_L,\sigma_R)\in\pi_{i,j}$. To get some intuition, a schematic version of this graph is given as Figure~\ref{fig:spanner-figure} in Appendix~\ref{app:directed}.

\subsection{Fractional Solution}\label{sec:directSol}
In this section, we provide a low-cost feasible vector solution for the $r$-th level Lasserre lift of the spanner instance described above.  Slightly more formally, we define values $\{y_{S}' : S\subseteq E\}$ and show that they form a feasible solution for the $r$-th level Lasserre lift $\text{SDP}_{Spanner}^r$, and show that the objective value is $O(|V|)$.  We do this by starting with a feasible solution  $\{y_\Psi^* : \Psi\subseteq (L\cup R)\times\Sigma\}$ to the $(r+2)$-th level Lasserre lift $\text{SDP}_{Proj}^{r+2}$ for the \textsc{Projection Games} instance (based on Lemma \ref{lem:projInstance}) we used to construct our directed spanner instance, and adapting it for the spanner context. Before defining $y_S'$, we define a function $\Phi:E\setminus E_{Outer}\rightarrow\mathscr{P}((L\cup R)\times\Sigma)$ (where $\mathscr{P}$ indicates the power set) as follows.
\[\Phi(e)=\begin{cases} 
\varnothing,&\mbox{if }e\in E_{LStars}\cup E_{M}\cup E_{RStars}\\
\{(c_i,\sigma)\},&\mbox{if }e\in E_L\mbox{ and }e\mbox{ has an endpoint }c_{i,\sigma}\in L_{labels}\\
\{(x_i,\sigma)\},&\mbox{if }e\in E_R\mbox{ and }e\mbox{ has an endpoint }x_{i,\sigma}\in R_{labels}\\
\end{cases}\]

We then extend the definition of $\Phi$ to $\mathscr{P}(E\setminus E_{Outer})\rightarrow\mathscr{P}((L\cup R)\times\Sigma)$ by setting $\Phi(S)=\cup_{e\in S}\Phi(e)$.

Next, we define the solution $\{y_{S}'\mid S\subseteq E\}$. For any set $S$ containing any edge in $E_{Outer}$, we define $y_{S}'=0$, otherwise, let $y_{S}'=y_{\Phi(S)}^*$. Note that based on how we defined the function $\Phi$, for all edges in $E_{LStars}\cup E_{M}\cup E_{RStars}$ we have $y'_S= y^*_\varnothing=1$. In other words, these edges will be picked integrally in our feasible solution. At a very high level, what we are doing is fractionally buying edges in $E_L,E_R$ and integrally buying edges in $E_M$ in order to span edges in $E_{Outer}$. The edges in $E_{LStars}$ and $E_{RStars}$ are used to span edges in $E_L$ and $E_R$. We first argue that our fractional solution has cost only $O(|V|)$. The proof is included in Appendix \ref{app:directSol}

\begin{lemma} \label{lem:directObj}
The objective value of $\mathbf{y'}$ is $O(|V|)$.
\end{lemma}

\subsubsection{Feasibility}\label{sec:directFeasibility}
In this section we show that the described vector solution is feasible for the $r$-th level of Lasserre, i.e., that all the moment matrices defined in $SDP_{Spanner}^r$ are PSD. This is the most technically complex part of the analysis, particularly for the slack moment matrices for edges in $E_{outer}$.  So we start with the easier matrices, working our way up to the more complicated ones.  In particular, we first use the fact that the base moment matrix in $SDP_{Proj}^{r+2}$ is PSD for the \textsc{Projection Games} solution $\mathbf{y}^*$ to show in Theorem~\ref{thm:main-moment} that the the base moment matrix of $SDP_{Spanner}^r$ is PSD for solution $\mathbf{y}'$.

\begin{theorem} \label{thm:main-moment}
The moment matrix $M_{r+2}(\mathbf{y}')=\left(y_{I\cup J}'\right)_{|I|,|J|\le r+2}$ is positive semidefinite, so does $M_{r+1}(\mathbf{y}')$.
\end{theorem}
\begin{proof}
We know that the moment matrix $M_{r+2}(\mathbf{y}^*)=\left(y_{\Psi_1\cup \Psi_2}^*\right)_{|\Psi_1|,|\Psi_2|\le r+2}\succcurlyeq0$, since $\mathbf{y}^*$ is a solution of $SDP_{Proj}^{r+2}$. Now for each principal submatrix $M$ of $M_{r+2}(\mathbf{y}')$, we consider three cases.  In the first case, suppose that $M$ has an index that contains an edge in $E_{Outer}$.  Then the whole row and whole column of this index is $0$, so the determinant is $0$. In the second case, $M$ includes two distinct indices $I,I'$ of $M_{r+2}(\mathbf{y}')$, such that $\Phi(I)=\Phi(I')$. In this case $M$ is not full rank, and thus the determinant is also $0$. Otherwise, $M$ does not include any indices that contain edges in $E_{Outer}$, and no two indices have the same $\Phi$ value.  Then $M$ is by definition a principal submatrix of $M_{r+2}(\mathbf{y}^*)$, since each row/column index of $M$ can be converted to a different index of $M_{r+2}(\mathbf{y}^*)$ based on function $\Phi$. Now, since all the principal submatrices of $M_{r+2}(\mathbf{y}')$ have non-negative determinant, by definition $M_{r+2}(\mathbf{y}')$ is PSD.
\end{proof}

Showing that the slack moment matrices of our spanner solution are all PSD is more subtle and requires a case by case analysis, combined with several properties of the Lasserre hierarchy. We divide this argument into three parts. First we show (in Theorem~\ref{thm:slack_ye=1}) that this is true for slack moment matrices corresponding to pairs $(u,v)$ for which we assigned $y_{(u,v)}'=1$. Then we show (Theorem~\ref{thm:slack:E_LE_R}) the same for edges in $E_L$ and $E_R$.  Finally, we handle the most difficult case of slack moment matrices corresponding to edges in $E_{outer}$ (Theorem~\ref{thm:Eouter}).  


\begin{theorem} \label{thm:slack_ye=1}
The slack moment matrix $M_r^\mathbf{z}(\mathbf{y}')=\left(\sum_{e\in E}z_ey_{I\cup J\cup\{e\}}'-y_{I\cup J}'\right)_{|I|,|J|\le r}$ is PSD for all $(u,v)\in E_{LStars}\cup E_{M}\cup E_{RStars}$ and $\mathbf{z}\in\mathcal{Z}^{u,v}$. 
\end{theorem}
\begin{proof}
Recall that for every $(u,v) \in E_{LStars}\cup E_{M}\cup E_{RStars}$ we set $y'_{(u,v)}=1$. So basic properties of Lasserre (Claim \ref{claim:=1}) imply that $y_{I \cup J \cup (u,v)}'= y_{I \cup J}'$ for all $|I|, |J| \leq r$. Thus
\begin{align*}
M_r^\mathbf{z}(\mathbf{y}')=&\left(\sum_{e\in E}z_ey_{I\cup J\cup\{e\}}'-y_{I\cup J}'\right)_{|I|,|J|\le r} =\left(\sum_{e\in E\setminus\{(u,v)\}}z_ey_{I\cup J\cup\{e\}}'+z_{(u,v)}y_{I\cup J\cup\{(u,v)\}}'-y_{I\cup J}'\right)_{|I|,|J|\le r}\\
=&\left(\sum_{e\in E\setminus\{(u,v)\}}z_ey_{I\cup J\cup\{e\}}'+1\cdot y_{I\cup J}'-y_{I\cup J}'\right)_{|I|,|J|\le r} =\sum_{e\in E\setminus P}z_e\left(y_{I\cup J\cup\{e\}}'\right)_{|I|,|J|\le r} \succcurlyeq 0
\end{align*}

Here the third equality follows from the fact $\mathbf{z}\in\mathcal{Z}^{u,v}$ and $(u,v)$ itself is a path connecting $u$ and $v$, and thus $z_{(u,v)}=1$. 
In the last equality we use the fact that $z_e\ge0$ for all $e\in E$, and $\left(y_{I\cup J\cup\{e\}}'\right)_{|I|,|J|\le r}$ is a principal submatrix of $M_{r+1}(\mathbf{y}')$, which is positive semidefinite according to Theorem \ref{thm:main-moment}.
\end{proof}

Now we prove a similar theorem for the edges in $E_L$ and $E_R$, which is a bit more complex since these edges are only bought fractionally in our solution.

\begin{theorem} \label{thm:slack:E_LE_R}
The slack moment matrix $M_r^\mathbf{z}(\mathbf{y}')=\left(\sum_{e\in E}z_ey_{I\cup J\cup\{e\}}'-y_{I\cup J}'\right)_{|I|,|J|\le r}$ is PSD for every $(u,v)\in E_L\cup E_R$ and $\mathbf{z}\in\mathcal{Z}^{u,v}$.
\end{theorem}
\begin{proof}
The proof for $(u,v)\in E_R$ is similar to the proof for $(u,v)\in E_L$, so without loss of generality we focus on the $E_L$ case. Recall that each edge in $E_L$ can be represented as $(c_i^l,c_{i,\sigma'})$, where $l \in [kK|\Sigma|]$. Now we consider two cases based on how the edge is spanned: $\sigma'=1$ or $\sigma'\ne1$.
When $\sigma'=1$, let $P_1=\{(c_i^l,c_{i,1})\}$, and for every $\sigma\in\Sigma\setminus\{1\}$, let $P_\sigma=\{(c_i^l,c_{i,\sigma}),(c_{i,\sigma},c_{i,1})\}$. When $\sigma'\ne1$, let $P_1=\{(c_i^l,c_{i,1}),(c_{i,1},c_{i,\sigma'})\}$, $P_{\sigma'}=\{(c_i^l,c_{i,\sigma'})\}$, and for every $\sigma\in\Sigma\setminus\{1,\sigma'\}$, let $P_\sigma=\{(c_i^l,c_{i,\sigma}),(c_{i,\sigma},c_{i,1}),(c_{i,1},c_{i,\sigma'})\}$. 
%

In the $\sigma'=1$ case, because paths $P_\sigma$ are disjoint, we can partition the slack moment matrix as following:
\begin{align*}
M_r^\mathbf{z}(\mathbf{y}')=&\left(\sum_{e\in E}z_ey_{I\cup J\cup\{e\}}'-y_{I\cup J}'\right)_{|I|,|J|\le r} =\left(\sum_{e\in E\setminus\bigcup_{\sigma\in\Sigma}P_\sigma}z_ey_{I\cup J\cup\{e\}}'+\sum_{\sigma\in\Sigma}\sum_{e\in P_\sigma}z_ey_{I\cup J\cup\{e\}}'-y_{I\cup J}'\right)_{|I|,|J|\le r}\\
=&\sum_{e\in E\setminus\bigcup_{\sigma\in\Sigma}P_\sigma}z_e\left(y_{I\cup J\cup\{e\}}'\right)_{|I|,|J|\le r} +\sum_{\sigma\in\Sigma}\sum_{e\in P_\sigma}z_e\left(y_{\Phi(I\cup J)\cup\Phi(e)}^*-y_{\Phi(I\cup J)\cup\{(c_i,\sigma)\}}^*\right)_{|I|,|J|\le r}\\
&+\sum_{\sigma\in\Sigma}\left(\sum_{e\in P_\sigma}z_e-1\right)\left(y_{\Phi(I\cup J)\cup\{(c_i,\sigma)\}}^*\right)_{|I|,|J|\le r} +\left(\sum_{\sigma\in\Sigma}y_{\Phi(I\cup J)\cup\{(c_i,\sigma)\}}^*-y_{\Phi(I\cup J)}^*\right)_{|I|,|J|\le r}
\end{align*}

The first term in the final sum is PSD because $z_e\ge0$ for all $e\in E$, and $\left(y_{I\cup J\cup\{e\}}'\right)_{|I|,|J|\le r}$ is a principal submatrix of $M_{r+1}(\mathbf{y}')$, which is positive semidefinite by Theorem \ref{thm:main-moment}.  The second term in the final sum is PSD because either $\Phi(e)=\{(c_i,\sigma)\}$ which makes the matrix a zero matrix, or $\Phi(e)=\varnothing$ and we can use Lemma \ref{lem:simple_slack} to prove the matrix is PSD. The third term in the final sum is PSD because $\sum_{e\in P}z_e\ge1$ for all $P \in \mathcal P_{u,v}$(since $\mathbf{z}\in\mathcal{Z}^{u,v}$) together with the fact that the matrix $\left(y_{\Phi(I\cup J)\cup\{(c_i,\sigma)\}}^*\right)_{|I|,|J|\le r}=\left(y_{I\cup J\cup\{(c_i^l,c_{i,\sigma})\}}'\right)_{|I|,|J|\le r}$ is a principal submatrix of $M_{r+1}(\mathbf{y}')$.  The fourth term in the final sum is the zero matrix because the slack moment constraint of $\text{SDP}_{Proj}^{r+2}$. Thus $M_r^\mathbf{z}(\mathbf{y}') \succcurlyeq 0$ when $\sigma' = 1$.

In the $\sigma'\ne1$ case, all the paths $P_\sigma, \sigma \in \Sigma \setminus \{\sigma'\}$ share the common edge $(c_{i,1},c_{i,\sigma'})$. Thus we can partition the sum in the slack moment matrix as follows:
\begin{align*}
&M_r^\mathbf{z}(\mathbf{y}')=\left(\sum_{e\in E}z_ey_{I\cup J\cup\{e\}}'-y_{I\cup J}'\right)_{|I|,|J|\le r} \\
=&\left(\sum_{e\in E\setminus\bigcup_{\sigma\in\Sigma}P_\sigma}z_ey_{I\cup J\cup\{e\}}'+\sum_{\sigma\in\Sigma}\sum_{e\in P_\sigma\setminus\{(c_{i,1},c_{i,\sigma'})\}}z_ey_{I\cup J\cup\{e\}}'+z_{(c_{i,1},c_{i,\sigma'})}y_{\Phi(I\cup J)\cup\varnothing}^*-y_{\Phi(I\cup J)}^*\right)_{|I|,|J|\le r}\\
=&\sum_{e\in E\setminus\bigcup_{\sigma\in\Sigma}P_\sigma}z_e\left(y_{I\cup J\cup\{e\}}'\right)_{|I|,|J|\le r} +\sum_{\sigma\in\Sigma}\sum_{e\in P_\sigma\setminus\{(c_{i,1},c_{i,\sigma'})\}}z_e\left(y_{\Phi(I\cup J)\cup\Phi(e)}^*-y_{\Phi(I\cup J)\cup\{(c_i,\sigma)\}}^*\right)_{|I|,|J|\le r}\\
&+\sum_{\sigma\in\Sigma}\left(\sum_{e\in P_\sigma\setminus\{(c_{i,1},c_{i,\sigma'})\}}z_e+z_{(c_{i,1},c_{i,\sigma'})}-1\right)\left(y_{\Phi(I\cup J)\cup\{(c_i,\sigma)\}}^*\right)_{|I|,|J|\le r}\\
&+(1-z_{(c_{i,1},c_{i,\sigma'})})\left(\sum_{\sigma\in\Sigma}y_{\Phi(I\cup J)\cup\{(c_i,\sigma)\}}^*-y_{\Phi(I\cup J)}^*\right)_{|I|,|J|\le r}
\end{align*}

The above matrix is PSD for the same reasons as in the $\sigma'=1$ case.
\end{proof}

Now we move to the main technical component of our integrality gap analysis: proving that the slack moment matrices  corresponding to outer edges are positive semidefinite. To show this we have to carefully partition the sum over different set of edges, and handle each part separately. We also need to manipulate these sums in a way that allows us to use the fact that slack moment matrix constraints defined on $\mathbf{y}^*$ for the Lasserre SDP relaxation of \textsc{Projection Games} are PSD. We first describe how the sum is partitioned and then prove the positive semidefiniteness of each piece. Note that in this argument we crucially use the fact that the SDP solution of \cite{CMMV17} satisfies all of the demands. 

\begin{theorem} \label{thm:Eouter}
The slack moment matrix $M_r^\mathbf{z}(\mathbf{y}')=\left(\sum_{e\in E}z_ey_{I\cup J\cup\{e\}}'-y_{I\cup J}'\right)_{|I|,|J|\le r}$ is PSD for every $(u,v)\in E_{outer}$ and $\mathbf{z}\in\mathcal{Z}^{u,v}$.
\end{theorem}
\begin{proof}
We first show how to decompose $M_r^\mathbf{z}(\mathbf{y}')$ as the sum of several simpler matrices. This will let us reason about each matrix differently based on the assigned values and their connection to the \textsc{Projection Games} constraints. We will then explain why each of these matrices is PSD. Observe that for each $(u,v)=(c_i^l,x_j^l)\in E_{outer}$, the set of stretch-$(2k-1)$ paths consist of the outer edge, or one of the paths that go through some labels $(\sigma_L,\sigma_R)$. It is not hard to see that any other path connecting such pairs has length larger than $(2k-1)$. More formally:
\begin{claim}\label{claim:path}
For every pair $(c_i^l,x_j^l)\in E_{Outer}$, the length $(2k-1)$ paths from $c_i^l$ to $x_j^l$ are:
\begin{itemize}
\item The path consisting of only the edge $(c_i^l,x_j^l)$.
\item The paths consisting of edges $\{(c_i^l,c_{i,\sigma_L})\}\cup E_M^{i,j,\sigma_L,\sigma_R}\cup\{(x_{j,\sigma_R},x_j^l)\}$ for some ${(\sigma_L,\sigma_R)\in\pi_{i,j}}$.
\end{itemize}
\end{claim}
We use this observation, and the fact that $y_e=0$ for all $e \in E_{Outer}$ to break the summation over $E$ in the definition of $M_r^\mathbf{z}(\mathbf{y}')$ into several pieces, which we will then argue are each PSD.  
\begin{align}
M_r^\mathbf{z}(\mathbf{y}')=&\left(\sum_{e\in E}z_ey_{I\cup J\cup\{e\}}'-y_{I\cup J}'\right)_{|I|,|J|\le r} =\sum_{e\in E}z_e\left(y_{\Phi(I\cup J)\cup\Phi(e)}^*\right)_{|I|,|J|\le r}-\left(y_{\Phi(I\cup J)}^*\right)_{|I|,|J|\le r}\nonumber\\
=&\sum_{\sigma_L\in\Sigma}z_{(c_i^l,c_{i,\sigma_L})}\left(y_{\Phi(I\cup J)\cup\{(c_i,\sigma_L)\}}^*\right)_{|I|,|J|\le r} +\sum_{(\sigma_L,\sigma_R)\in\pi_{i,j}}\sum_{e\in E_M^{i,j,\sigma_L,\sigma_R}}z_e\left(y_{\Phi(I\cup J)\cup\varnothing}^*\right)_{|I|,|J|\le r}\nonumber\\
&+\sum_{\sigma_R\in\Sigma}z_{(x_{j,\sigma_R},x_j^l)}\left(y_{\Phi(I\cup J)\cup\{(x_j,\sigma_R)\}}^*\right)_{|I|,|J|\le r} +\sum_{e\in E\setminus(E_L^{i,l}\cup E_M^{i,j}\cup E_R^{j,l})}z_e\left(y_{I\cup J\cup\{e\}}'\right)_{|I|,|J|\le r}\nonumber\\
&-\left(y_{\Phi(I\cup J)}^*\right)_{|I|,|J|\le r}\nonumber\\
=&\sum_{\sigma_L\in\Sigma}z_{(c_i^l,c_{i,\sigma_L})}\left(y_{\Phi(I\cup J)\cup\{(c_i,\sigma_L)\}}^*-\sum_{\sigma_R:(\sigma_L,\sigma_R)\in\pi_{i,j}}y_{\Phi(I\cup J)\cup\{(c_i,\sigma_L),(x_j,\sigma_R)\}}^*\right)_{|I|,|J|\le r}\label{eqn:ML}\\
&+\sum_{(\sigma_L,\sigma_R)\in\pi_{i,j}}\sum_{e\in E_M^{i,j,\sigma_L,\sigma_R}}z_e\left(y_{\Phi(I\cup J)}^*-y_{\Phi(I\cup J)\cup\{(c_i,\sigma_L),(x_j,\sigma_R)\}}^*\right)_{|I|,|J|\le r}\label{eqn:MM}\\
&+\sum_{\sigma_R\in\Sigma}z_{(x_{j,\sigma_R},x_j^l)}\left(y_{\Phi(I\cup J)\cup\{(x_j,\sigma_R)\}}^*-\sum_{\sigma_L:(\sigma_L,\sigma_R)\in\pi_{i,j}}y_{\Phi(I\cup J)\cup\{(c_i,\sigma_L),(x_j,\sigma_R)\}}^*\right)_{|I|,|J|\le r}\label{eqn:MR}\\
&+\sum_{e\in E\setminus(E_L^{i,l}\cup E_M^{i,j}\cup E_R^{j,l})}z_e\left(y_{I\cup J\cup\{e\}}'\right)_{|I|,|J|\le r}\label{eqn:Mrest}\\
&-\left(y_{\Phi(I\cup J)}^*-\sum_{(\sigma_L,\sigma_R)\in\pi_{i,j}}y^*_{\Phi(I\cup J)\cup\{(c_i,\sigma_L),(x_j,\sigma_R)\}}\right)_{|I|,|J|\le r}\label{eqn:Msum}\\
&+\sum_{(\sigma_L,\sigma_R)\in\pi_{i,j}}\left(z_{(c_i^l,c_{i,\sigma_L})}+\sum_{e\in E_M^{i,j,\sigma_L,\sigma_R}}z_e+z_{(x_{j,\sigma_R},x_j^l)}-1\right)\label{eqn:Zsum}\\
&\quad\times\left(y_{\Phi(I\cup J)\cup\{(c_i,\sigma_L),(x_j,\sigma_R)\}}^*\right)_{|I|,|J|\le r} \label{eqn:MLR}
\end{align}

Note that to prove the third equality we have subtracted and added the sum over $(\sigma_L,\sigma_R) \in \pi_{i,j}$, and then partitioned it over the other sums. Now we argue why each of the above matrices is PSD.

The matrix in \eqref{eqn:MM} is PSD by applying Lemma \ref{lem:simple_slack} twice.  The matrix in \eqref{eqn:Mrest} is PSD because it is a principal submatrix of $M_{r+1}(\mathbf{y}')$, which is positive semidefinite by Theorem \ref{thm:main-moment}.  The matrix in \eqref{eqn:MLR} is PSD because it equals $\left(y_{I\cup J\cup\{(c_i^l,c_{i,\sigma_L}),(x_{j,\sigma_R},x_j^l)\}}'\right)_{|I|,|J|\le r}$, and is also a principal submatrix of $M_{r+1}(\mathbf{y}')$.  Its coefficient in \eqref{eqn:Zsum} is non-negative because  $\sum_{e\in P}z_e\ge1$ for all path $P\in\mathcal{P}_{c_i^l,x_j^l}$ by the definition of $\mathbf{z}$.

We now argue that matrices in \eqref{eqn:ML}, \eqref{eqn:MR}, and \eqref{eqn:Msum} are all-zero matrices, which will complete the proof.   In order to show this, we need the following claim, which uses the fact that the fractional solution to the \textsc{Projection Games} instance satisfies all of the edges.   

\begin{claim}\label{claim:pair=0}
$y_{\Psi\cup\{(c_i,\sigma_L),(x_j,\sigma_R)\}}^*=0$ for all $\{c_i,x_j\}\in E_{Proj}$, $(\sigma_L,\sigma_R)\notin\pi_{i,j}$, and $|\Psi|\le2r$.
\end{claim}
\begin{proof}
By Lemma \ref{lem:projInstance}, we know that $\sum_{(\sigma_L,\sigma_R)\in\pi_{i,j}}y_{(c_i,\sigma_L),(x_j,\sigma_R)}^*=1$.
By the slack moment constraints in $\text{SDP}_{Proj}^r$, we know that $M_r^{c_i}(\mathbf{y}^*)$ and $M_r^{x_j}(\mathbf{y}^*)$ are both all-zero matrices. Hence,
\begin{align*}
&\sum_{\sigma_L\in\Sigma}\sum_{\sigma_R\in\Sigma}y_{(c_i,\sigma_L),(x_j,\sigma_R)}^* = \sum_{\sigma_R\in\Sigma}y_{(x_j,\sigma_R)}^* = 1 =\sum_{(\sigma_L,\sigma_R)\in\pi_{i,j}}y_{(c_i,\sigma_L),(x_j,\sigma_R)}^*
\end{align*}

In other words, since $y_{(c_i,\sigma_L),(x_j,\sigma_R)}^*\ge0$ for all $\sigma_L\in\Sigma$ and $\sigma_R\in\Sigma$, it follows for all $(\sigma_L,\sigma_R)\notin\pi_{i,j}$ that $y_{(c_i,\sigma_L),(x_j,\sigma_R)}^*=0$. Then Claim \ref{claim:=0} implies $y_{\Psi\cup\{(c_i,\sigma_L),(x_j,\sigma_R)\}}^*=0$.
\end{proof}

Next, we argue that all entries of the matrix in line \eqref{eqn:ML} are zero. For any entry with index $I$ and $J$ we have
\begin{align*}
y_{\Phi(I\cup J)\cup\{(c_i\sigma_L)\}}^*=&\sum_{\sigma_R\in\Sigma}y_{\Phi(I\cup J)\cup\{(c_i,\sigma_L),(x_j,\sigma_R)\}}^*\ = \sum_{\sigma_R:(\sigma_L,\sigma_R)\in\pi_{i,j}}y_{\Phi(I\cup J)\cup\{(c_i,\sigma_L),(x_j,\sigma_R)\}}^* = 0,
\end{align*}
where the last equality follows from Claim \ref{claim:pair=0}.

Similarly, we argue that matrix in line \eqref{eqn:Msum} is all-zero. For any entry with index $I$ and $J$,
\begin{align*}
y_{\Phi(I\cup J)}^*=&\sum_{\sigma_L\in\Sigma}y_{\Phi(I\cup J)\cup\{(c_i,\sigma_L)\}} = \sum_{\sigma_L\in\Sigma}\sum_{\sigma_R\in\Sigma}y_{\Phi(I\cup J)\cup\{(c_i,\sigma_L),(x_j,\sigma_R)\}}^* =\sum_{(\sigma_L,\sigma_R)\in\pi_{i,j}}y_{\Phi(I\cup J)\cup\{(c_i,\sigma_L),(x_j,\sigma_R)\}}^* = 0.
\end{align*}
Again, we have used Claim \ref{claim:pair=0} in the last equality.

Finally, an argument similar to what we used for \eqref{eqn:ML} implies that the matrix in line \eqref{eqn:MR} is also all-zero, proving the theorem. 
\end{proof}

\subsection{Integral Solutions}
In this section, we argue that any $(2k-1)$-spanner of the graph $G = (V, E)$ we constructed in Section~\ref{sec:dir_instance} needs to have many edges.  More precisely, we prove the following lemma.

\begin{lemma}\label{lem:directSound}
The optimal $(2k-1)$-spanner of $G$ has at least $nkK|\Sigma|\sqrt{K}$ edges.
\end{lemma}

We argue that the size of the optimal spanner must be at least $nkK|\Sigma|\sqrt{K}$, otherwise the \textsc{Projection Games} instance has a solution in which $\omega(\frac{m\ln n}{\varepsilon})$ edges are satisfied. This contradicts Lemma~\ref{lem:projInstance}. We first use the following claim, for which the proof can be found in Appendix \ref{app:directOuter}.
\begin{claim}\label{claim:directOuter}
Any $(2k-1)$-spanner $S$ of $G$ can be transformed to another $(2k-1)$-spanner $S'$ such that $S'\cap E_{Outer}=\varnothing$ and $|S'|\le3|S|$.
\end{claim} 

\paragraph{Proof of Lemma \ref{lem:directSound}:}
If the optimal solution of the $(2k-1)$-spanner instance is less than $nkK|\Sigma|\sqrt{K}$, then by Claim \ref{claim:directOuter} there is a solution $S'$ with less than $3nkK|\Sigma|\sqrt{K}$ edges that does not use any edge from $E_{Outer}$.  Then by a simple averaging argument, there must exist some $l\in[kK|\Sigma|]$ such that $S'\cap E_L^l$ and $S'\cap E_R^l$ both have size less than $3n\sqrt{K}$.  Since each vertex in $E_L\cup E_R$ corresponds to a pair in $(L\cup R)\times\Sigma$, the set $S'\cap(E_L^l\cup E_R^l)$ corresponds to a label assignment $\Psi=\Phi(S'\cap(E_L^l\cup E_R^l))$, where each vertex may be assigned multiple labels in $\Psi$.  It is easy to see that $\Psi$ satisfies all the edges in $E_{Proj}$. This is because for each edge $\{c_i,x_j\}\in E_{Proj}$, there is an outer edge $(c_i^l,x_j^l)$, and so there is a length $(2k-1)$ path in $S'$ from $c_i^l$ to $x_j^l$. From Claim \ref{claim:path} we know that this corresponds to a label $\sigma_L$ on $c_i$ and a label $\sigma_R$ on $x_j$, and this satisfies edge $(c_i,x_j)$.

Now we will define another assignment $\Psi'$ that satisfies fewer edges, but where each vertex in $L$ has at most $\frac{12\sqrt{K}}{n^\varepsilon}$ labels, and each vertex in $R$ has at most $24\sqrt{K}$ labels.  We just ignore the vertices in $L$ that have degree larger than $\frac{12\sqrt{K}}{n^\varepsilon}$ labels, that is, we let their corresponding edges remain unsatisfied. There are at most $\frac{3n\sqrt{K}}{\frac{12\sqrt{K}}{n^\varepsilon}}=\frac{n^{1+\varepsilon}}{4}$ such vertices. We also ignore the vertices in $R$ which have larger than $24\sqrt{K}$ labels, and similarly there are at most $\frac{3n\sqrt{K}}{24\sqrt{K}}=\frac{n}{8}$ of them. The number of edges that are still satisfied by $\Psi'$ is at least $|L|K-\frac{n^{1+\varepsilon}}{4}\cdot K-\frac{n}{8}\cdot 2Kn^\varepsilon=\frac{mK}{2}$, since the vertices in $L$ have degree $K$ and the vertices in $R$ have degree at most $2Kn^\varepsilon$ (by Lemma \ref{lem:projInstance}).

The final step is to find an assignment that satisfies $\omega(\frac{m\ln n}{\varepsilon})$ edges, where each vertex has only one label. This can be achieved by a probabilistic argument. If we randomly choose a label from the existing labels for each vertex in $\Psi'$, then the probability that each edge is still satisfied is $\frac{n^\varepsilon}{12\sqrt{K}}\cdot\frac{1}{24\sqrt{K}}=\frac{n^\varepsilon}{288K}$. Thus in expectation there will be $\frac{n^\varepsilon}{288K}\cdot\frac{mK}{2}=\frac{mn^\varepsilon}{576}=\omega(\frac{m\ln n}{\varepsilon})$ edges satisfied. Which means there must exist one assignment that satisfies $\omega(\frac{m\ln n}{\varepsilon})$ edges. This contradicts  Lemma \ref{lem:projInstance}, finishing the proof.
\qed

\subsection{Proof of Theorem~\ref{thm:directSpanner}}
From Lemma \ref{lem:directObj} we know the solution $\mathbf{y}'$ has objective value $O(V)$. From Theorem \ref{thm:main-moment}, \ref{thm:slack_ye=1}, \ref{thm:slack:E_LE_R}, and Lemma \ref{lem:projInstance} we know that $\mathbf{y}'$ is a feasible solution to the $|V|^{\Omega(\varepsilon)}$-th level Lasserre SDP $SDP_{Spanner}^r$. From Lemma \ref{lem:directSound} we know that the optimal solution is at least $nkK|\Sigma|\sqrt{K}$. Therefore, the inegrality gap is $\frac{nkK|\Sigma|\sqrt{K}}{mkK|\Sigma|}=\left(\frac{|V|}{k}\right)^{\frac{1}{18}-\Theta(\varepsilon)}$.

\section{Lasserre Integrality Gap for Undirected $(2k-1)$-Spanner}

In order to extend these techniques to the undirected case, we need to make a number of changes.  First, for technical reasons we need to replace the middle path $E_M$ with edges, and instead add ``outside" paths (as was done in~\cite{DZ16}).  More importantly, though, we have the same fundamental problem that always arises when moving from directed to undirected spanners: without directions on edges, there can be many more short cycles (and thus ways of spanning edges) in the resulting graph.  In particular, if we directly change every edge on the integrality gap instance of \textsc{Directed $(2k-1)$-Spanner} in the previous section to be undirected, then Claim \ref{claim:path} no longer holds.  There can be other ways of spanning outer edges, e.g., it might be possible to span an outer edges using only other outer edges.  
Claim \ref{claim:directOuter} is affected for the same reason.

This difficulty is fundamentally caused because the graph of the \textsc{Projection Games} instance from~\cite{CMMV17} that we use as our starting point might have short cycles.  These turn into short cycles of outer edges in our spanner instance.  In order to get around this, we first carefully do subsampling and pruning to remove a select subset of edges in $E_{Proj}$, causing the remaining graph to have  large girth (at least $2k+2$) but without losing too much of its density or any of the other properties that we need. This is similar to what was done by~\cite{DZ16} to prove an integrality gap for the base LP, but here we are forced to start with the instance of~\cite{CMMV17}, which has far more complicated structure than the random \textsc{Unique Games} instances used by~\cite{DZ16}.  This is the main technical difficulty, but once we overcome it we can use the same ideas as in Section~\ref{sec:directed} to prove Theorem \ref{thm:undirectSpanner}.  Details can be found in Appendix \ref{app:undirected}.

\bibliographystyle{plain}
\bibliography{refs}

\begin{thebibliography}{10}

\bibitem{ADDJS93}
Ingo Alth\"{o}fer, Gautam Das, David Dobkin, Deborah Joseph, and Jos\'{e}
  Soares.
\newblock On sparse spanners of weighted graphs.
\newblock {\em Discrete Comput. Geom.}, 9(1):81--100, 1993.

\bibitem{BDZ18}
Amy Babay, Michael Dinitz, and Zeyu Zhang.
\newblock Characterizing demand graphs for (fixed-parameter) shallow-light
  steiner network.
\newblock {\em arXiv preprint arXiv:1802.10566}, 2018.

\bibitem{BRS11}
Boaz Barak, Prasad Raghavendra, and David Steurer.
\newblock Rounding semidefinite programming hierarchies via global correlation.
\newblock In {\em Proceedings of the 2011 IEEE 52Nd Annual Symposium on
  Foundations of Computer Science}, FOCS '11, pages 472--481, Washington, DC,
  USA, 2011. IEEE Computer Society.

\bibitem{BBMRY13}
Piotr Berman, Arnab Bhattacharyya, Konstantin Makarychev, Sofya Raskhodnikova,
  and Grigory Yaroslavtsev.
\newblock Approximation algorithms for spanner problems and directed steiner
  forest.
\newblock {\em Inf. Comput.}, 222:93--107, 2013.

\bibitem{bhaskara2012}
Aditya Bhaskara, Moses Charikar, Venkatesan Guruswami, Aravindan
  Vijayaraghavan, and Yuan Zhou.
\newblock Polynomial integrality gaps for strong sdp relaxations of densest
  k-subgraph.
\newblock In {\em Proceedings of the twenty-third annual ACM-SIAM symposium on
  Discrete Algorithms}, pages 388--405. SIAM, 2012.

\bibitem{BGJRW09}
Arnab Bhattacharyya, Elena Grigorescu, Kyomin Jung, Sofya Raskhodnikova, and
  David~P Woodruff.
\newblock Transitive-closure spanners.
\newblock {\em SIAM Journal on Computing}, 41(6):1380--1425, 2012.

\bibitem{CMM06}
Moses Charikar, Konstantin Makarychev, and Yury Makarychev.
\newblock Near-optimal algorithms for unique games.
\newblock In {\em Proceedings of the Thirty-eighth Annual ACM Symposium on
  Theory of Computing}, STOC '06, pages 205--214, New York, NY, USA, 2006. ACM.

\bibitem{chlamtac2007}
Eden Chlamtac.
\newblock Approximation algorithms using hierarchies of semidefinite
  programming relaxations.
\newblock In {\em Foundations of Computer Science, 2007. FOCS'07. 48th Annual
  IEEE Symposium on}, pages 691--701. IEEE, 2007.

\bibitem{CD16}
Eden Chlamt{\'a}{\v{c}} and Michael Dinitz.
\newblock Lowest-degree $k$-spanner: Approximation and hardness.
\newblock {\em Theory of Computing}, 12(15):1--29, 2016.

\bibitem{CMMV17}
Eden Chlamt{\'{a}}c, Pasin Manurangsi, Dana Moshkovitz, and Aravindan
  Vijayaraghavan.
\newblock Approximation algorithms for label cover and the log-density
  threshold.
\newblock In Philip~N. Klein, editor, {\em Proceedings of the Twenty-Eighth
  Annual {ACM-SIAM} Symposium on Discrete Algorithms, {SODA} 2017}, pages
  900--919. {SIAM}, 2017.

\bibitem{chlamtac2008}
Eden Chlamtac and Gyanit Singh.
\newblock Improved approximation guarantees through higher levels of sdp
  hierarchies.
\newblock In {\em Approximation, Randomization and Combinatorial Optimization.
  Algorithms and Techniques}, pages 49--62. Springer, 2008.

\bibitem{CDKL17}
Eden Chlamt\'a\v{c}, Michael Dinitz, Guy Kortsarz, and Bundit Laekhanukit.
\newblock Approximating spanners and directed steiner forest: Upper and lower
  bounds.
\newblock In {\em Proceedings of the 28th Annual ACM-SIAM Symposium on Discrete
  Algorithms}, SODA '17, 2017.

\bibitem{CDK12}
Eden Chlamt\'a\v{c}, Michael Dinitz, and Robert Krauthgamer.
\newblock Everywhere-sparse spanners via dense subgraphs.
\newblock In {\em Proceedings of the 53rd IEEE Annual Symposium on Foundations
  of Computer Science}, FOCS '12, pages 758--767, 2012.

\bibitem{DKR16}
Michael Dinitz, Guy Kortsarz, and Ran Raz.
\newblock Label cover instances with large girth and the hardness of
  approximating basic k-spanner.
\newblock {\em ACM Trans. Algorithms}, 12(2):25:1--25:16, December 2015.

\bibitem{DK11}
Michael Dinitz and Robert Krauthgamer.
\newblock Directed spanners via flow-based linear programs.
\newblock In {\em Proceedings of the Forty-third Annual ACM Symposium on Theory
  of Computing}, STOC '11, pages 323--332, 2011.

\bibitem{DK11-FT}
Michael Dinitz and Robert Krauthgamer.
\newblock Fault-tolerant spanners: Better and simpler.
\newblock In {\em Proceedings of the 30th Annual ACM SIGACT-SIGOPS Symposium on
  Principles of Distributed Computing}, PODC '11, pages 169--178, 2011.

\bibitem{DZ16}
Michael Dinitz and Zeyu Zhang.
\newblock Approximating low-stretch spanners.
\newblock In {\em Proceedings of the 27th Annual ACM-SIAM Symposium on Discrete
  Algorithms}, SODA '16, 2016.

\bibitem{EP07}
Michael Elkin and David Peleg.
\newblock The hardness of approximating spanner problems.
\newblock {\em Theory Comput. Syst.}, 41(4):691--729, 2007.

\bibitem{Erd64}
Paul Erd\H{o}s.
\newblock Extremal problems in graph theory.
\newblock In {\em Theory Graphs Appl., Proc. Symp. Smolenice 1963}, pages
  29--36, 1964.

\bibitem{friedrichs18}
Stephan Friedrichs and Christoph Lenzen.
\newblock Parallel metric tree embedding based on an algebraic view on
  moore-bellman-ford.
\newblock {\em Journal of the ACM (JACM)}, 65(6):43, 2018.

\bibitem{friggstad2014}
Zachary Friggstad, Jochen K{\"o}nemann, Young Kun-Ko, Anand Louis, Mohammad
  Shadravan, and Madhur Tulsiani.
\newblock Linear programming hierarchies suffice for directed steiner tree.
\newblock In {\em International Conference on Integer Programming and
  Combinatorial Optimization}, pages 285--296. Springer, 2014.

\bibitem{Kor01}
Guy Kortsarz.
\newblock On the hardness of approximating spanners.
\newblock {\em Algorithmica}, 30(3):432--450, 2001.

\bibitem{lasserre2001}
Jean~B Lasserre.
\newblock An explicit exact sdp relaxation for nonlinear 0-1 programs.
\newblock In {\em International Conference on Integer Programming and
  Combinatorial Optimization}, pages 293--303. Springer, 2001.

\bibitem{laurent2003}
Monique Laurent.
\newblock A comparison of the sherali-adams, lov{\'a}sz-schrijver, and lasserre
  relaxations for 0--1 programming.
\newblock {\em Mathematics of Operations Research}, 28(3):470--496, 2003.

\bibitem{lovasz1991}
L{\'a}szl{\'o} Lov{\'a}sz and Alexander Schrijver.
\newblock Cones of matrices and set-functions and 0--1 optimization.
\newblock {\em SIAM journal on optimization}, 1(2):166--190, 1991.

\bibitem{manurangsi2015}
Pasin Manurangsi.
\newblock {\em On approximating projection games}.
\newblock PhD thesis, Massachusetts Institute of Technology, 2015.

\bibitem{PS89}
David Peleg and Alejandro~A. Sch{\"a}ffer.
\newblock Graph spanners.
\newblock {\em Journal of Graph Theory}, 13(1):99--116, 1989.

\bibitem{PU89}
David Peleg and Jeffrey~D. Ullman.
\newblock An optimal synchronizer for the hypercube.
\newblock {\em SIAM J. Comput.}, 18(4):740--747, 1989.

\bibitem{rothvoss2011}
Thomas Rothvo{\ss}.
\newblock Directed steiner tree and the lasserre hierarchy.
\newblock {\em arXiv preprint arXiv:1111.5473}, 2011.

\bibitem{rothvoss2013}
Thomas Rothvo{\ss}.
\newblock The lasserre hierarchy in approximation algorithms.
\newblock {\em Lecture Notes for the MAPSP}, pages 1--25, 2013.

\bibitem{schoenebeck2008}
Grant Schoenebeck.
\newblock Linear level lasserre lower bounds for certain k-csps.
\newblock In {\em 2008 49th Annual IEEE Symposium on Foundations of Computer
  Science}, pages 593--602. IEEE, 2008.

\bibitem{sherali1990}
Hanif~D Sherali and Warren~P Adams.
\newblock A hierarchy of relaxations between the continuous and convex hull
  representations for zero-one programming problems.
\newblock {\em SIAM Journal on Discrete Mathematics}, 3(3):411--430, 1990.

\bibitem{TZ01}
Mikkel Thorup and Uri Zwick.
\newblock Compact routing schemes.
\newblock In {\em Proceedings of the thirteenth annual ACM symposium on
  Parallel algorithms and architectures}, pages 1--10. ACM, 2001.

\bibitem{tulsiani2009}
Madhur Tulsiani.
\newblock Csp gaps and reductions in the lasserre hierarchy.
\newblock In {\em Proceedings of the forty-first annual ACM symposium on Theory
  of computing}, pages 303--312. ACM, 2009.

\end{thebibliography}

\appendix

\section{Proofs about Properties of the Lasserre Hierarchy} \label{app:lasserre}

\begin{proof}[Proof of Claim~\ref{claim:=1}]
Consider the determinant of the following principal submatrix of $M_r(y)$:
\[\begin{vmatrix}
y_\varnothing&y_I&y_J\\
y_I&y_I&y_{I\cup J}\\
y_J&y_{J\cup I}&y_J
\end{vmatrix}=-(y_{I\cup J}-y_I)^2\ge0\]
Thus $y_{I\cup J}=y_J$.
\end{proof}

\begin{proof}[Proof of Claim~\ref{claim:=0}]
Consider the determinant of principal submatrix of $M_t(y)$ indexed by sets $I$ and $J$:
\[\begin{vmatrix}
y_I&y_{I\cup J}\\
y_{J\cup I}&y_J
\end{vmatrix}=-y_{I\cup J}^2\ge0\]
Thus $y_{I\cup J}=0$.
\end{proof}

\begin{proof}[Proof of Lemma~\ref{lem:simple_slack}]
This follows from the fact that $L_r(K) \subseteq N^+_r(L_{r-1}(K))$, where $N^+_r$ is the SDP variant of Lovasz-Schrijver hierarchy. For a formal proof see \cite{laurent2003}.
\end{proof}

\section{Proofs from Section~\ref{sec:projection}}

\subsection{Equivalence of SDPs}\label{app:projEquiv}

The level-$r$ Lasserre SDP for \textsc{Projection Games} as stated in \cite{CMMV17} is written in terms of inner products of vector variables.  For every $\Psi\subseteq(L\cup R)\times\Sigma$, let $U_\Psi$ be a vector in an appropriately high-dimensional vector space. 
The constraints will force $U_\Psi=0$ if $\Psi$ assigns two different labels to the same vertex. If there exists $(v,\sigma),(v,\sigma')\in\Psi$ and $\sigma\ne\sigma'$, we say $\Psi$ is \textit{inconsistent}.  The following SDP is the one analyzed by~\cite{CMMV17}.
\[\begin{array}{rll}
&\hspace{-3em}\text{SDP}_{CMMV}^r:&\\
\max&\sum\limits_{(v_L,v_R)\in E,(\sigma_L,\sigma_r)\in\pi_{(u,v)}}\|U_{(v_L,\sigma_L),(v_R,\sigma_R)}\|^2\hspace{-7em}&\\
s.t.&\|U_\varnothing\|^2=1&\\
&\langle U_{\Psi_1},U_{\Psi_2}\rangle=0&\forall|\Psi_1|,|\Psi_2|\le r,\Psi_1\cup\Psi_2\mbox{ is inconsistent}\\
&\langle U_{\Psi_1},U_{\Psi_2}\rangle\ge0&\forall|\Psi_1|,|\Psi_2|\le r\\
&\langle U_{\Psi_1},U_{\Psi_2}\rangle=\langle U_{\Psi_3},U_{\Psi_4}\rangle&\forall|\Psi_1|,|\Psi_2|,|\Psi_3|,|\Psi_4|\le r,\Psi_1\cup\Psi_2=\Psi_3\cup\Psi_4\\
&\sum\limits_{\sigma\in\Sigma} \|U_{(v,\sigma)}\|^2=1 & \forall v\in V\\
\end{array}\]

It can by shown that $\text{SDP}_{CMMV}^r$ is equivalent to our formulation $\text{SDP}_{Proj}^r$.  Intuitively, if $U_\Psi$ is a solution to $\text{SDP}_{CMMV}^r$, then $y_\Psi=\|U_\Psi\|^2$ is a solution to $\text{SDP}_{Proj}^r$ with the same objective value. Also, if $y_\Psi$ is a solution to $\text{SDP}_{Proj}^r$, then consider the Cholesky decomposition $M_r(y)=LL^T$. The rows of $L$ can be the solution $U_\Psi$ to $\text{SDP}_{CMMV}^r$ with the same objective value.

We will only prove the first direction, because in our paper we just need all the constraints in $\text{SDP}_{Proj}^r$ to be satisfied. We start with the following claim:


\begin{claim}\label{claim:projConsist}
Let $\{U_\Psi\mid\Psi\subseteq(L\cup R)\times\Sigma\}$ be a solution that satisfies all the constraints in $\text{SDP}_{CMMV}^r$. Then for all $v\in V$, $\sum\limits_{\sigma\in\Sigma}U_{(v,\sigma)}=U_\varnothing$.
\end{claim}
\begin{proof}
We will show that $\|\sum\limits_{\sigma\in\Sigma}U_{(v,\sigma)}-U_\varnothing\|^2=1-\sum\limits_{\sigma\in\Sigma}\|U_{(v,\sigma)}\|^2=0$.
\begin{align}
\|\sum_{\sigma\in\Sigma}U_{(v,\sigma)}-U_\varnothing\|^2&=\|\sum_{\sigma\in\Sigma}U_{(v,\sigma)}\|^2-2\langle \sum_{\sigma\in\Sigma}U_{(v,\sigma)}, U_\varnothing\rangle+\|U_\varnothing\|^2\\ 
&=\sum_{\sigma\in\Sigma} \|U_{(v,\sigma)}\|^2-2\sum_{\sigma\in\Sigma} \langle U_{(v,\sigma)}, U_\varnothing\rangle+\|U_\varnothing\|^2\\
&=\sum_{\sigma\in\Sigma} \|U_{(v,\sigma)}\|^2-2\sum_{\sigma\in\Sigma} \|U_{(v,\sigma)}\|^2+\|U_\varnothing\|^2\\
&=1-\sum_{\sigma\in\Sigma} \|U_{(v,\sigma)}\|^2
\end{align}
Note that in line (2) we used the fact that $\forall \sigma_1,\sigma_2\in\Sigma: \langle U_{(v,\sigma_1)}, U_{(v,\sigma_2)} \rangle=0$, and thus by the Pythagorean theorem we can move the sum out of the norm.
\end{proof}

This gives us the following lemma:

\begin{lemma} \label{lem:projConsist}
Let $\{U_\Psi\mid\Psi\subseteq(L\cup R)\times\Sigma\}$ be a solution that satisfies all the constraints in $\text{SDP}_{CMMV}^r$. Then for each $\Psi\subseteq (L\cup R)\times\Sigma$ where $|\Psi|\le r$ we have $\sum_{\sigma\in\Sigma}\|U_{\Psi\cup\{(v,\sigma)\}}\|^2=\|U_\Psi\|^2.$
\end{lemma}
\begin{proof}
\begin{align*}
\sum_{\sigma\in\Sigma}\|U_{\Psi\cup\{(v,\sigma)\}}\|^2-\|U_\Psi\|^2&=\sum_{\sigma\in\Sigma}\langle U_\Psi, U_{(v,\sigma)}\rangle-\langle U_\Psi , U_\varnothing\rangle \\
=\sum_{\sigma\in\Sigma} \langle U_\Psi, U_{(v,\sigma)}\rangle-\langle U_\Psi ,U_\varnothing \rangle  &= \langle U_\Psi,\sum_{\sigma\in\Sigma} U_{(v,\sigma)}\rangle-\langle U_\Psi , U_\varnothing\rangle \\
&=\langle U_\Psi , \sum_{\sigma\in\Sigma} U_{(v,\sigma)}-U_\varnothing\rangle=0
\end{align*}
The last equality is from Claim \ref{claim:projConsist}.
\end{proof}


Now we can prove the main lemma in this section:

\begin{lemma}\label{lem:projMoment}
Given a solution $\{U_\Psi\mid\Psi\subseteq(L\cup R)\times\Sigma\}$ to $\text{SDP}_{CMMV}^r$, then $\{y_\Psi=\|U_\Psi\|^2\mid\Psi\subseteq(L\cup R)\times\Sigma\}$ is feasible for $\text{SDP}_{Proj}^r$, and the objective values are the same.
\end{lemma}
\begin{proof}
The objective values are clearly the same by definition of $y$.

Because a symmetric matrix $M$ is positive semidefinite if and only if it can be written as inner products of $n$ vectors (i.e. there exists $U_1,\mathellipsis,U_n$ such that $M=\left(\langle U_i,U_j\rangle\right)_{i,j\in[n]}$), we know that the moment matrix $M_r(y)=\left(\langle U_{\Psi_1},U_{\Psi_2}\rangle\right)_{|\Psi_1|,|\Psi_2|\le r}$ is positive semidefinite.

For the slack moment matrix $M_r^v(y)$, consider any entry $\sum_{\sigma\in\Sigma}y_{\Psi_1\cup\Psi_2\cup\{(v,\sigma)\}}-y_{\Psi_1\cup\Psi_2}$ in it. From Lemma \ref{lem:projConsist} and the definition of $y$ we know that it equals to zero.
\end{proof}


\subsection{Instance in \cite{CMMV17}}\label{app:projInstance}

In this section we describe the instance of \textsc{Projection Games} used in~\cite{CMMV17} to prove their integrality gap. This instance is based on $k$-CSP instances of \cite{tulsiani2009}, but we will describe the \textsc{Projection Games} instance directly since we do not need to use any extra properties of the $k$-CSP instance. Note that the parameters are slightly different as in \cite{CMMV17} and \cite{manurangsi2015}, for example $D$ rather than $D-1$ and $\varepsilon$ rather than $\rho$.

Let $R=(x_1,\mathellipsis,x_n)$ and $L=(c_1,\mathellipsis,c_m)$ with $m=n^{1+\varepsilon}$, where $x$ represents variables and $c$ represents constraints.

Let $C\subseteq\mathbb{F}_q^{q-1}$ be a linear code of distance $q-D$, length $q-1$ and dimension $D$, where $q$ is a prime equals to $n^\frac{1-\varepsilon}{5}$ and $D=3$ (Reed-Solomon code achieve this). Here, $0<\varepsilon<1$ is a small constant.

Let $\Sigma=[q^D]$, then the alphabet size is $|\Sigma|=|C|=q^D$.

The edges and the projection are created randomly. Each vertex $c_i$ will choose $K=q-1$ neighbors $T_i=\{x_{i_1},\mathellipsis,x_{i_K}\}$ in $R$, and a \emph{shift} vector $b_i=(b_{i,1},\mathellipsis,b_{i,K})\in\mathbb{F}_q^K$ both uniformly at random.

The projection $\pi_{(c_i,x_{i_j})}=\pi_{i,i_j}$ is defined as 
\[\{(\sigma_L,\sigma_R)\mid(\sigma_R+b_{i,j})\mbox{ is the }j\mbox{-th coordinate of the }\sigma_L\text{-th code in }C\}\]

This is a projection from $\Sigma=[q^D]$ to $[q]$. Where each $\sigma\in[q]$ has $q^{D-1}$ preimages, each $\sigma\in\Sigma\setminus[q]$ has $0$ preimage.

In \cite{CMMV17} a feasible vector solution to $\text{SDP}_{CMMV}^r$ was proposed that has several properties that we need to use in our analysis. One the crucial property is that it is a perfect solution, meaning all projections can be satisfied. Roughly speaking, this property allows us to use this solution for Min-Rep, which is the minimization variant of this problem. More formally, 
\begin{lemma}[\cite{CMMV17}]\label{lem:projComplete}
For the above instance, with probability at least $1-o(1)$, there exists a feasible solution $U^*$ to the $r=N^{\Omega(\varepsilon)}$-th level $\text{SDP}_{CMMV}^r$, such that ${\sum\limits_{(\sigma_L,\sigma_R)\in\pi_{(c_i,x_j)}} \|U_{(c_i,\sigma_L),(x_j,\sigma_R)}\|^2=1}$ for all $\{c_i,x_j\}\in E_{Proj}$.\
\end{lemma}

Directly from Lemma \ref{lem:projMoment}, we have:

\begin{corollary}\label{cor:projComplete}
For the above instance, with probability at least $1-o(1)$, there exist a feasible solution $\mathbf{y}^*$ for the $r=N^{\Omega(\varepsilon)}$-th level $\text{SDP}_{Proj}^r$, such that ${\sum\limits_{(\sigma_L,\sigma_R)\in\pi_{(c_i,x_j)}}y_{(c_i,\sigma_L),(x_j,\sigma_R)}=1}$ for all $\{c_i,x_j\}\in E_{Proj}$.
\end{corollary}

\begin{lemma}[\cite{CMMV17}] \label{lem:projSound}
With probability at least $1-o(1)$, at most $O\left(\frac{n^{1+\varepsilon}\ln n}{\varepsilon}\right)$ edges can be satisfied.
\end{lemma}

Besides the properties in \cite{CMMV17}, we need another property to be satisfied, and we will show that it holds with probability $1-o(1)$.

\begin{claim}\label{claim:projDegree}
With probability at least $1-o(1)$, the largest degree in $R$ is at most $2Kn^\varepsilon$.
\end{claim}
\begin{proof}
For each vertex $v\in R$, the probability that each left vertex $u$ chooses the edge connecting $v$ is $\frac{K}{n}$, and these events are independent. Let $\Delta(v)$ be the degree of $v$, then $\mathbb{E}[\Delta(v)] = Kn^\varepsilon$. Using Chernoff bound, we have
\[\Pr[\Delta(v)>2Kn^\varepsilon] \le e^{-\frac{\mathbb{E}[\Delta(v)]}{3}}=e^{-\frac{Kn^\varepsilon}{3}}=o\left(\frac{1}{n^2}\right)\]

The claim will then follow using a union bound over all $n$ vertices in $R$. 
\end{proof}

Using union bound over Corollary \ref{cor:projComplete}, Lemma \ref{lem:projSound}, and Claim \ref{claim:projDegree}, we have proved Lemma \ref{lem:projInstance}.

\section{Proofs from Section~\ref{sec:directed}} \label{app:directed}

\subsection{Proofs from Section~\ref{sec:spanner-lp}} \label{app:spanner-lp}

To prove Theorem~\ref{thm:spanner-lp-equivalent}, we prove two claims, one for each direction.  

\begin{claim} \label{claim:limitedLP_dir1}
Given a feasible solution $\mathbf{x}^*$ to $\text{LP}_{Spanner}$, we can find $f_P^*$ for each $(u,v)\in E$ and $P\in\mathcal{P}_{u,v}$, which makes $(\mathbf{x}^*,\mathbf{f}^*)$ a feasible solution to $\text{LP}_{Spanner}^{Flow}$.
\end{claim}
\begin{proof}
For each $(u,v) \in E$, let $\mathbf{z}^* \in \mathcal{Z}^{u,v}$ be the vector that minimizes $\sum\limits_{e\in E}x_e^*z_e$. In other words $\mathbf{z}^*$ is an optimal solution to the following LP:
\[\begin{array}{rll}
&\hspace{-3em}\text{LP}_{Cut}^{u,v}:&\\
\min&\sum\limits_{e\in E}x_e^*z_e &\\
s.t.&\sum\limits_{e\in P}z_e \ge1 &\forall P\in \mathcal{P}_{u,v}\\
&z_e\ge0&\forall e\in E\\
\end{array}\]
Since $\mathbf{x}^*$ is a feasible solution to $\text{LP}_{Spanner}$, we have that $\sum\limits_{e\in E}x_e^*z^*_e\ge1$. Now let us consider the dual of the above LP.
\[\begin{array}{rll}
&\hspace{-3em}\text{LP}_{Flow}^{u,v}:&\\
\max&\sum\limits_{p\in \mathcal{P}_{u,v}} f_p&\\
s.t.&\sum\limits_{P\in \mathcal{P}_{u,v}:e\in P}f_P\le x_e^*&\forall(u,v)\in E,\forall e\in E\\
&f_P\ge0&\forall(u,v)\in E, P\in \mathcal{P}_{u,v}
\end{array}\]
Let $\{f_P^*\mid P\in\mathcal{P}_{u,v}\}$ be the optimal solution for $\text{LP}_{Flow}^{u,v}$. We argue that if we combine all $f_P^*$ solutions accroding to every $(u,v)\in E$, then the solution $\mathbf{(x,f)}$ is also feasible for $\text{LP}_{Spanner}^{Flow}$. Note that by strong duality we have $max\sum\limits_{p\in \mathcal{P}_{u,v}} f_p = \min \sum\limits_{e\in E}x_e^*z_e \geq 1$. This means that both the flow and capacity constraints of $\text{LP}_{Spanner}^{Flow}$ are satisfied.
\end{proof}

\begin{claim}
Given a feasible solution $(\mathbf{x}^*,\mathbf{f}^*)$ to $\text{LP}_{Spanner}^{Flow}$, then $\mathbf{x}^*$ is a feasible solution to $\text{LP}_{Spanner}$.
\end{claim}
\begin{proof}
This follows from a very similar argument to the previous claim. This time we start with the $(\mathbf{x}^*,\mathbf{f}^*)$ solution to $\text{LP}_{Spanner}^{Flow}$, where part of the solution is also a feasible solution to each $\text{LP}_{Flow}^{u,v}$ defined in proof of Claim \ref{claim:limitedLP_dir1} for every $(u,v)\in E$. We use duality again and now note that by flow constraints of $\text{LP}_{Spanner}^{Flow}$ we get $1 \le\max\sum\limits_{p\in \mathcal{P}_{u,v}} f_p = \min \sum\limits_{e\in E}x_e^* z_e$.
\end{proof}

Note that as defined there are actually an infinite number of constraints (one for each vector in $\mathcal Z_{u,v}$, for each $(u,v) \in E$).  However, via the convexity of $\mathcal Z_{u,v}$, it is sufficient to only have constraints for the extreme points of the polytopes characterized by each $\mathcal{Z}^{u,v}$.  Thus $\text{LP}_{Spanner}$ is indeed an LP.

As a slight aside, it is worth noting that the ``antispanner" LP of~\cite{BBMRY13} is precisely $\text{LP}_{Spanner}$ but where the constraints are only for \emph{integral} $\mathbf{z} \in \mathcal Z_{u,v}$.  This immediately implies that their antispanner LP is no stronger than $\text{LP}_{Spanner}$, so our integrality gaps will also hold for lifts of the antispanner LP.

\subsection{Directed $(2k-1)$-Spanner Instance}

In Figure \ref{fig:spanner-figure} we present the instance described in Section \ref{sec:dir_instance}.
\begin{figure}
\begin{center}
\includegraphics[height=70mm, width=100mm,]{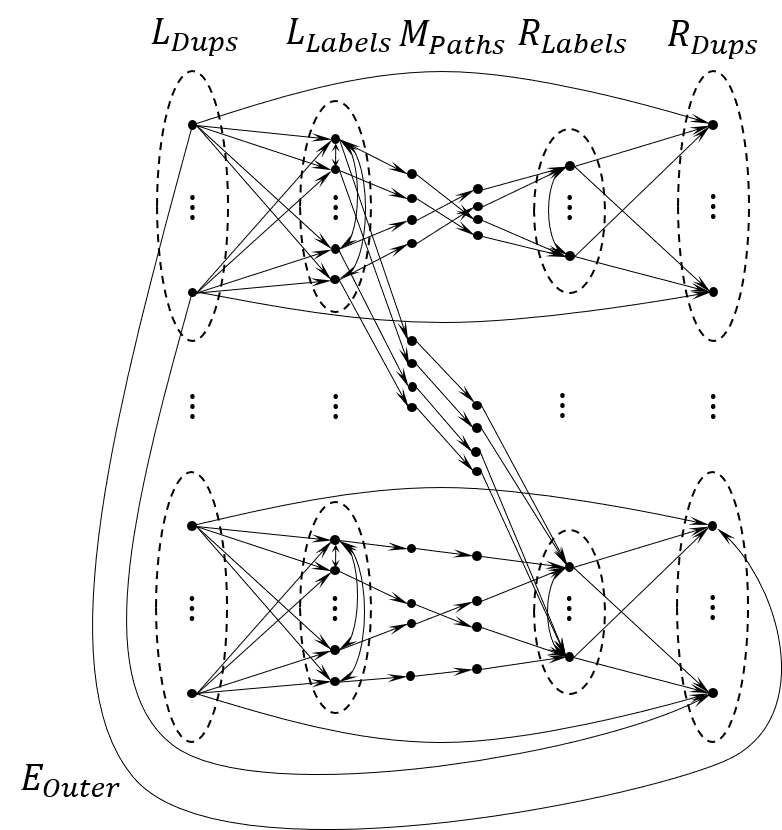}
\end{center}
\caption{\textsc{Directed $(2k-1)$-Spanner} instance}
\label{fig:spanner-figure}
\end{figure}

\subsection{Proof of Claim~\ref{lem:directObj}}\label{app:directSol}

\begin{proof}
First, we count the number of edges that are defined to be $y_e'=y_\varnothing^*=1$. Recall that these are the edges in $E_{LStars}\cup E_{M}\cup E_{RStars}$, of which there are
$|L|\cdot(2|\Sigma|-2)+(2k-3)\cdot|\Sigma|\cdot K|L|+|R|\cdot(2|\Sigma|-2)=O(|V|)$.

Now we compute the cost contributed by the other edges $e$ where $y_e' \neq 0$.  These are precisely the edges in $E_L$ and $E_R$, and so we have that

\[\sum_{e\in E_L}y_e'+\sum_{e\in E_R}y_e'=kK|\Sigma|\cdot\sum_{i\in|L|}\sum_{\sigma\in\Sigma}y_{(c_i,\sigma)}^*+kK|\Sigma|\cdot\sum_{i\in|R|}\sum_{\sigma\in\Sigma} y_{(x_i,\sigma)}^*=kK|\Sigma||L|+kK|\Sigma||R|=O(|V|).\]

Thus the total cost of $\mathbf{y'}$ is $O(|V|)$ as claimed.  
\end{proof}

\subsection{Proof of Claim \ref{claim:directOuter}}\label{app:directOuter}
\begin{claim}
Any $(2k-1)$-spanner $S$ of $G$ can be transformed to another $(2k-1)$-spanner $S'$ such that $S'\cap E_{Outer}=\varnothing$ and $|S'|\le3|S|$.
\end{claim} 
\begin{proof} 
For $k>2$, every edge in $E_M$ must be in $S$, because there is no other way to span these edges.  For any outer edge $(u,v)\in E_{Outer} \cap S$, we can replace $(u,v)$ by an arbitrary path of length $(2k-1)$ from $u$ to $v$ that only uses one edge from $E_L$, one edge from $E_R$ and edges from $E_M$ (see, e.g., Claim \ref{claim:path}). Clearly each time we replace an outer edge in this way, we will introduce at most $3$ new edges into the solution: one in $E_L$, one in $E_R$, and at most one in $E_M$ (if $k=2$).  Therefore, this will at most add a factor of $3$ to the total size, and the new solution does not contain any edge in $E_{Outer}$.
\end{proof}

\section{Lasserre Integrality Gap for Undirected $(2k-1)$-Spanner} \label{app:undirected}
In this section we focus on the undirected $(2k-1)$-spanner. Before discussing the spanner instance, we need to make some modifications to the \textsc{Projection Games} instance to get a better bound. We discuss this modification and relevant properties of the instance in Section \ref{sec:modifiedProj}. Then we will describe the undirected spanner instance, which has basically the same structure proposed by \cite{DK11} and is slightly different from the directed $(2k-1)$-spanner instance.

\subsection{\textsc{Projection Games} Modifications for Undirected Spanner} \label{sec:modifiedProj}

We first perform some modifications to the \textsc{Projection Games} instance to increase the girth while keepings other desired properties in the feasible and optimal integral solutions.
For $(2k-1)$ spanner, we first change $K$ (i.e. the degree of the vertices in $L$) from $n^{\frac{1-\varepsilon}{5}}-1$ to $\min\left\{n^{\frac{1-\varepsilon}{5}}-1,n^\frac{1-\varepsilon}{2k-1}\right\}$, so that each vertex in $L$ will randomly choose fewer neighbors in $R$. This will let us bound the number of cycles in this instance. More formally,

\begin{claim}
With probability at least $1-o(1)$, the number of cycles of length at most $2k$ in the instance is at most $\frac{n^{1+\varepsilon}K}{2}$.
\end{claim}
\begin{proof}
Note that the instance graph is a bipartite graph that is $K$-regular from the left. We first show that in such a random graph the expected number of cycles of length $2\ell\le2k$ is at most $o\left(\frac{n^{1+\varepsilon} K}{\ell^2}\right)$. For computing this expectation, we first count the sets $L_1 \cup L_2$, where $L_1 \subseteq L$ and $L_2 \subseteq R$ that can form a cycle of length $2\ell$. Since $|L_1|=|L_2|= \ell$, the number of such sets is ${n^{1+\varepsilon}\choose \ell} \cdot {n \choose \ell}$.

Now for counting the cycles we first one of the $\ell$ nodes in $L_2$, and after fixing this, there are $(\ell-1)!\ell!$ possibilities for ordering of other nodes. Since cycles can be counted by $2$ directions, we divide the number by $2$. Thus the number of possible cycles between $L$ and $R$ is:
\[{n^{1+\varepsilon}\choose\ell}\cdot{n\choose\ell}\cdot\frac{(\ell-1)!\ell!}{2}\le n^{4\ell+2\varepsilon\ell}\]

Then we note that the probability that two specific edges are chosen for an specific vertex in $L$ is $\frac{{n^{1+ \varepsilon} \choose K-2}}{{n^{1+\varepsilon} \choose K}}$. Hence the expected number of cycles of size $2\ell$ is:
\begin{align*}
 E[\#\text{cycles of length }2\ell]&\le n^{4\ell+2\varepsilon\ell}\cdot\left(\frac{K(K-1)}{(n^{1+\varepsilon}-K-1)(n^{1+\varepsilon}-K)}\right)^\ell\\
&\le K^{2l}n^{-2\varepsilon\ell}\left(1+\frac{2K}{n^{1+\varepsilon}-2K}\right)^{2\ell}\\
&\le K^{2l}n^{-2\varepsilon\ell}e^\frac{4K\ell}{n^{1+\varepsilon}-2K}\le n^\frac{(2\ell-1)(1-\varepsilon)}{2k-1}K\cdot n^{-2\varepsilon\ell}\cdot O(1)\\
&= n^{1+\varepsilon}K\cdot o\left(\frac{1}{\ell^2}\right)=o\left(\frac{n^{1+\varepsilon} K}{\ell^2}\right)
 \end{align*}
We then sum over all $1\le\ell\le k$, and get that the expected number of $2k$-cycles is at most $n^{1+\varepsilon}K\cdot o\left(\sum_{\ell=1}^k\frac{1}{\ell^2}\right)=n^{1+\varepsilon}K\cdot o(1)$. Then with Markov inequality the claim follows. 
\end{proof}

Next, we remove one edge from each cycles of length at most $2k$, so that the girth will become $(2k+2)$. Note that with probability $1-o(1)$, there are still $\frac{n^{1+\varepsilon}K}{2}$ edges left.

It is easy to see that Corollary \ref{cor:projComplete} still holds on the new instance for the same $\mathbf{y}^*$. Lemma \ref{lem:projSound} also holds because the new instance is just a subgraph of the previous instance. Claim \ref{claim:projDegree} holds for the same reason. We summarize the properties of the obtained instance in the following lemma.
\begin{lemma}\label{lem:projModified}
For any small constant $0<\varepsilon<1$, there exists a \textsc{Projection Games} instance noted $(L,R,E_{Proj}^k,\Sigma,(\pi_e)_{e\in E_{Proj}})$ which has the following properties:
\begin{enumerate}
\item $\Sigma=[n^\frac{3-3\varepsilon}{5}]$, $|E_{Proj}^k|\ge\frac{n^{1+\varepsilon}K}{2}$, $R=\{x_1,\mathellipsis,x_n\}$, $L=\{c_1,\mathellipsis,c_m\}$, where $m=n^{1+\varepsilon}$.
\item There exists a feasible solution $\mathbf{y}^*$ for the $r=n^{\Omega(\varepsilon)}$-th level of Lasserre for the Undirected $(2k-1)$-spanner Lasserre, such that ${\sum\limits_{(\sigma_L,\sigma_R)\in\pi_{(c_i,x_j)}}y_{(c_i,\sigma_L),(x_j,\sigma_R)}=1}$ for all $\{c_i,x_j\}\in E_{Proj}^k$.
\item At most $O\left(\frac{n^{1+\varepsilon}\ln n}{\varepsilon}\right)$ edges can be satisfied.
\item The degree of vertices in $L$ is $K=\min\left\{n^{\frac{1-\varepsilon}{5}}-1,n^\frac{1-\varepsilon}{2k-1}\right\}$, and the degree of vertices in $R$ is at most $2Kn^\varepsilon$.
\item The girth of graph $(L\cup R,E_{Proj}^k)$ is at least $2k+2$.
\end{enumerate}
\end{lemma}

\subsection{Spanner Instance}
The spanner instance used here is the instance of \cite{DK11}, which as stated earlier, is slightly different from the instance used in our integrality gap for the directed $k$-spanner. Given a modified \textsc{Projection Games} instance $(L=\cup_{i\in m}\{c_i\},R=\cup_{i\in n}\{x_i\},E_{Proj}^k,\Sigma,(\pi_e)_{e\in E_{Proj}^k})$ from Lemma \ref{lem:projModified}, we construct an undirected $(2k-1)$-spanner instance $G=(V,E)$ as follows (Note that $K$ is the degree of the vertices in $L$, before removing edges in cycles):

For every $c_i\in L$, we create $|\Sigma|+K|\Sigma|(k-1)$ vertices: $c_{i,\sigma}$ for all $\sigma\in\Sigma$ and $c_i^{(j,l)}$ for all $j\in[k-1],l\in[K|\Sigma|]$.  
 We also create edges $\{c_i^{(1,l)},c_{i,\sigma}\}$ for each $\sigma\in\Sigma$ and $l\in[K|\Sigma|]$. Also $\{c_i^{(j,l)},c_i^{(j+1,l)}\}$ for each $j\in[k-2]$ and $l\in[K|\Sigma|]$. In other words, we have a set of $K|\Sigma|$ paths and we denote the $l$-th path as $E_{LPath}^{i,l}$. There is a complete bipartite graph, between the starting points of these paths and the vertices $c_{i,\sigma},\sigma\in\Sigma$.

For every $x_i\in R$, we create $|\Sigma|+K|\Sigma|(k-1)$ vertices: $x_{i,\sigma}$ for $\sigma\in\Sigma$ and $x_i^{(j,l)}$ for $j\in[k-1],l\in[K|\Sigma|]$. We also create edge $\{x_i^{(1,l)},x_{i,\sigma}\}$ for each $\sigma\in\Sigma$ and $l\in[K|\Sigma|]$. Also $\{x_i^{(j,l)},x_i^{(j+1,l)}\}$ for each $j\in[k-2]$ and $l\in[K|\Sigma|]$. In other words, we have a set of $K|\Sigma|$ paths and we denote the $l$-th path as $E_{RPath}^{i,l}$. There is a complete bipartite graph, between the starting points of these paths and the vertices $x_{i,\sigma},\sigma\in\Sigma$. 

For each $e=\{c_i,x_j\}\in E_{Proj}^k$, we add edges $\{c_i^{(k-1,l)},x_j^{(k-1,l)}\}$ for each $l\in[K|\Sigma|]$, and add edges $\{c_{i,\sigma_L},x_{j,\sigma_R}\}$ for $(\sigma_L,\sigma_R)\in\pi_{i,j}$. In other words, we have $K|\Sigma|$ duplicates of $E_{Proj}^k$ that connects further endpoints of the left and right outside paths, which we call \textit{outer} edges $E_{Outer}$. We also have $|\Sigma|$ duplicates of $E_{Proj}^k$ between $L_{Labels}$ and $R_{Labels}$, these are the edges in $E_M$.

Similar to the Section \ref{sec:dir_instance}, we also need some other edges, $E_{LStars}$ and $E_{LStars}$ inside groups of $L_{Labels}$ and $R_{Labels}$.

To be more specific, $V=L_{Paths}\cup L_{Labels}\cup R_{Labels}\cup R_{Paths}$, $E=E_{LPaths}\cup E_L\cup E_{LStars}\cup E_M\cup E_{RStars}\cup E_R\cup E_{RPaths}\cup E_{Outer}$, such that:
\[\begin{array}{lll}
L_{Labels}=\{c_{i,\sigma}\mid i\in|L|,\sigma\in\Sigma\},&\hspace{-6em}L_{Paths}=\{c_i^{(j,l)}\mid i\in|L|,l\in[K|\Sigma|],j\in[k-1]\},\\
R_{Labels}=\{x_{i,\sigma}\mid i\in|R|,\sigma\in\Sigma\},&\hspace{-6em}R_{Paths}=\{x_i^{(j,l)}\mid i\in|R|,l\in[K|\Sigma|],j\in[k-1]\},\\
E_{LPath}^{i,l}=\{\{c_i^{(j,l)},c_i^{(j+1,l)}\}\mid i\in|L|,l\in[K|\Sigma|],j\in[k-2]\},&E_{LPaths}=\cup_{i\in|L|}\cup_{l\in[K|\Sigma|]}E_{LPath}^{i,l},\\
E_{RPath}^{i,l}=\{\{x_i^{(j,l)},x_i^{(j+1,l)}\}\mid i\in|R|,l\in[K|\Sigma|],j\in[k-2]\},&E_{RPaths}=\cup_{i\in|R|}\cup_{l\in[K|\Sigma|]}E_{RPath}^{i,l},\\
E_L^{i,l}=\{\{c_i^{(1,l)},c_{i,\sigma}\}\mid\sigma\in\Sigma\},&\hspace{-6em}E_L^l=\cup_{i\in|L|}E_L^{i,l},&\hspace{-9em}E_L=\cup_{l\in[K|\Sigma|]}E_L^l,\\
E_R^{i,l}=\{\{x_i^{(1,l)},x_{i,\sigma}\}\mid\sigma\in\Sigma\},&\hspace{-6em}E_R^l=\cup_{i\in|R|}E_R^{i,l},&\hspace{-9em}E_R=\cup_{l\in[K|\Sigma|]}E_R^l,\\
E_M^{i,j}=\{\{c_{i,\sigma_L},x_{j,\sigma_R}\}\mid(\sigma_L,\sigma_R)\in\pi_{i,j}\},&\hspace{-6em}E_M=\cup_{i,j:\{c_i,x_j\}\in E_{Proj}^k}E_M^{i,j}\\
E_{Outer}=\{\{c_i^{(k-1,l)},x_j^{(k-1,l)}\}\mid\{c_i,x_j\}\in E_{Proj},l\in[K|\Sigma|]\}\\
E_{LStars}=\{\{c_{i,1},c_{i,\sigma}\}\mid i\in|L|,\sigma\in\Sigma\setminus\{1\}\},&\hspace{-6em}E_{RStars}=\{\{x_{i,1},x_{i,\sigma}\}\mid i\in|R|,\sigma\in\Sigma\setminus\{1\}\},
\end{array}\]

\begin{figure}
\begin{center}
\includegraphics[scale=0.7]{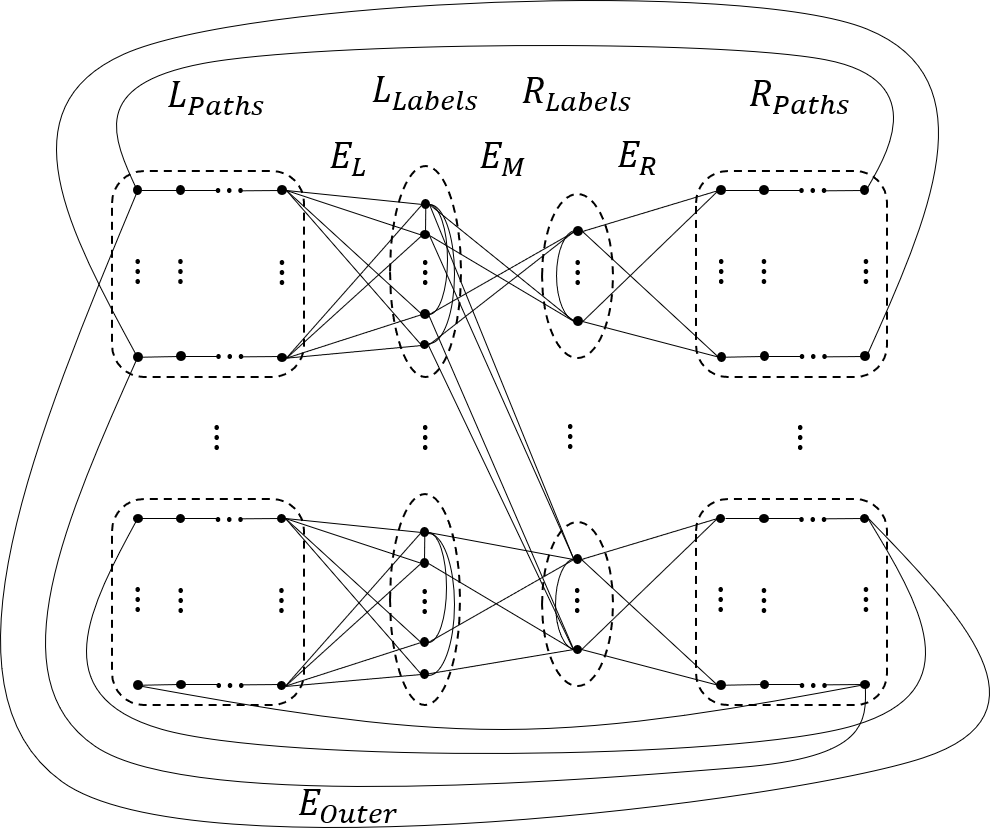} \label{fig:undir-spanner}
\end{center}
\caption{\textsc{Basic $(2k-1)$-Spanner} instance}
\end{figure}

\subsection{Fractional Solution}

Given the solution $\{y_\Psi^*\mid\Psi\subseteq (L\cup R)\times\Sigma\}$ of the $(r+2)$-th level Lasserre SDP in Lemma \ref{lem:projComplete}, we define $\Phi(e)$ similar to Section \ref{sec:directSol}:
\[\Phi(e)=\begin{cases} 
\varnothing,&\mbox{if }e\in E_{LPaths}\cup E_{LStars}\cup E_M\cup E_{RStars}\cup E_{RPaths}\\
\{(c_i,\sigma)\},&\mbox{if }e\in E_L\mbox{ and }e\mbox{ has an endpoint }c_{i,\sigma}\in L_{labels}\\
\{(x_i,\sigma)\},&\mbox{if }e\in E_R\mbox{ and }e\mbox{ has an endpoint }x_{i,\sigma}\in R_{labels}\\
\end{cases}\]
And $\Phi$ can be extended to $\mathscr{P}(E\setminus E_{Outer})\rightarrow\mathscr{P}((L\cup R)\times\Sigma)$ by letting $\Phi(S)=\cup_{e\in S}\Phi(e)$.

Similarly, if $S$ contain any edge in $E_{Outer}$, we define $y_{S}'=0$. Otherwise, let $y_{S}'=y_{\Phi(S)}^*$.

\begin{lemma}
The objective value of $\textbf{y}$ is $O(|V|)$.
\end{lemma}
\begin{proof}
We first count the number of edges which is defined to be $y_e'=y_\varnothing^*=1$. That is the number of edges in the set $E_{LPaths}\cup E_{LStars}\cup E_{M}\cup E_{RStars}\cup E_{RPaths}$, which is:
\[|L|\cdot(k-2)\cdot K|\Sigma|+|L|(|\Sigma|-1)+\frac{n^{1+\varepsilon}K}{2}\cdot|\Sigma|+|R|(|\Sigma|-1)+|R|\cdot(k-2)\cdot K|\Sigma|=O(|V|)\]

For other edges with non-zero value of $y_e'$ we have,
\[\sum_{e\in E_L}y_e'+\sum_{e\in E_R}y_e'=K|\Sigma|\cdot\sum_{i\in|L|}\sum_{\sigma\in\Sigma}y_{(c_i,\sigma)}^*+K|\Sigma|\cdot\sum_{i\in|R|}\sum_{\sigma\in\Sigma} y_{(x_i,\sigma)}^*=K|\Sigma||L|+K|\Sigma||R|=O(|V|)\]

Therefore,
\[\sum_{e\in E}y_e'=\sum_{e\in E_{Outer}}y_e'+\sum_{e\in E^1}y_e'+\sum_{e\in E_L}y_e'+\sum_{e\in E_R}y_e'=0+O(V)+O(V)=O(V)\]
\end{proof}

\subsection{Feasibility}

The proof for the moment matrix $M_{r+1}(\mathbf{y}')=\left(y_{I\cup J}'\right)_{|I|,|J|\le r+1}\succcurlyeq0$ and the slack moment matrix $M_r^\mathbf{z}(\mathbf{y}')=\left(\sum_{e\in E}z_ey_{I\cup J\cup\{e\}}'-y_{I\cup J}'\right)_{|I|,|J|\le r}\succcurlyeq0$ where $\mathbf{z}\in\mathcal{Z}^{u,v}$ and $\{u,v\}\in E\setminus E_{Outer}$ is exactly the same as Section \ref{sec:directFeasibility}. The arguments for showing that the slack moment matrices are PSD is also similar to Section \ref{sec:directFeasibility}, except that the calculations must be done on the the new instance.

\begin{theorem}
The slack moment matrix $M_r^\mathbf{z}(\mathbf{y}')=\left(\sum_{e\in E}z_ey_{I\cup J\cup\{e\}}'-y_{I\cup J}'\right)_{|I|,|J|\le r}\succcurlyeq0$ is PSD for every $\{u,v\}\in E_{outer}$ and $\mathbf{z}\in\mathcal{Z}^{u,v}$.
\end{theorem} 

For each $\{u,v\}=\{c_i^{(k-1,l)},x_j^{(k-1,l)}\}\in E_{outer}$, the set of stretch $(2k-1)$-paths consist of the outer edge, or one of the paths that go through some labels $(\sigma_L,\sigma_R)$. Using the fact that the girth is $(2k+2)$ and after modifications of the spanner instance we can make a structural claim similar to \ref{claim:path}. We have,
\begin{claim}\label{claim:undirectedPath}
For every pair $\{c_i^{(k-1,l)},x_j^{(k-1,l)}\}\in E_{Outer}$, the length $(2k-1)$ paths connecting $c_i^{(k-1,l)}$ and $x_j^{(k-1,l)}$ are:
\begin{itemize}
\item The path consisted of only the edge $\{c_i^{(k-1,l)},x_j^{(k-1,l)}\}$.
\item A path consisted of edges in $E_{LPath}^{i,l}\cup\{\{c_i^{(1,l)},c_{i,\sigma_L}\},\{c_{i,\sigma_L}, x_{j,\sigma_R}\},\{x_{j,\sigma_R},x_j^{(1,l)}\}\}\cup E_{RPath}^{j,l}$ for some ${(\sigma_L,\sigma_R)\in\pi_{i,j}}$.
\end{itemize}
\end{claim}

Therefore, other than edge  $\{c_i^{(k-1,l)},x_j^{(k-1,l)}\}$ itself, all the edges that may appear in these paths are in the set $E_{LPath}^{i,l}\cup E_L^{i,l}\cup E_M^{i,j}\cup E_R^{j,l}\cup E_{RPath}^{i,l}$.

Thus we can break the summation over $E$ in to several pieces, and have the following equation:

\begin{align}
M_r^\mathbf{z}(\mathbf{y}')=&\left(\sum_{e\in E}z_ey_{I\cup J\cup\{e\}}'-y_{I\cup J}'\right)_{|I|,|J|\le r}\nonumber\\
=&\sum_{e\in E}z_e\left(y_{\Phi(I\cup J)\cup\Phi(e)}^*\right)_{|I|,|J|\le r}-\left(y_{\Phi(I\cup J)}^*\right)_{|I|,|J|\le r}\nonumber\\
=&\sum_{e\in E_{LPath}^{i,l}\cup E_{RPath}^{i,l}}z_e\left(y_{\Phi(I\cup J)\cup\varnothing}^*\right)_{|I|,|J|\le r}\nonumber\\
&+\sum_{\sigma_L\in\Sigma}z_{\{c_i^{(1,l)},c_{i,\sigma_L}\}}\left(y_{\Phi(I\cup J)\cup\{(c_i\sigma_L)\}}^*\right)_{|I|,|J|\le r}\nonumber\\
&+\sum_{(\sigma_L,\sigma_R)\in\pi_{i,j}}z_{\{c_{i,\sigma_L}, x_{j,\sigma_R}\}}\left(y_{\Phi(I\cup J)\cup\varnothing}^*\right)_{|I|,|J|\le r}\nonumber\\
&+\sum_{\sigma_R\in\Sigma}z_{\{x_{j,\sigma_R},x_j^{(1,l)}\}}\left(y_{\Phi(I\cup J)\cup\{(x_j,\sigma_R)\}}^*\right)_{|I|,|J|\le r}\nonumber\\
&+\sum_{e\in E\setminus(E_{LPath}^{i,l}\cup E_L^{i,l}\cup E_M^{i,j}\cup E_R^{j,l}\cup E_{RPath}^{i,l})}z_e\left(y_{I\cup J\cup\{e\}}'\right)_{|I|,|J|\le r}\nonumber\\
&-\left(y_{\Phi(I\cup J)}^*\right)_{|I|,|J|\le r}\nonumber\\
=&\sum_{e\in E_{LPath}^{i,l}\cup E_{RPath}^{i,l}}z_e\left(y_{\Phi(I\cup J)}^*-\sum_{(\sigma_L,\sigma_R)\in\pi_{i,j}}y_{\Phi(I\cup J)\cup\{(c_i,\sigma_L)\}\cup\{(x_j,\sigma_R)\}}^*\right)_{|I|,|J|\le r}\label{eqn:Mpaths}\\
&+\sum_{\sigma_L\in\Sigma}z_{\{c_i^{(1,l)},c_{i,\sigma_L}\}}\left(y_{\Phi(I\cup J)\cup\{(c_i\sigma_L)\}}^*-\sum_{\sigma_R:(\sigma_L,\sigma_R)\in\pi_{i,j}}y_{\Phi(I\cup J)\cup\{(c_i,\sigma_L)\}\cup\{(x_j,\sigma_R)\}}^*\right)_{|I|,|J|\le r}\label{eqn:MLU}\\
&+\sum_{(\sigma_L,\sigma_R)\in\pi_{i,j}}z_{\{c_{i,\sigma_L}, x_{j,\sigma_R}\}}\left(y_{\Phi(I\cup J)\cup\varnothing}^*-y_{\Phi(I\cup J)\cup\{(c_i,\sigma_L)\}\cup\{(x_j,\sigma_R)\}}^*\right)_{|I|,|J|\le r}\label{eqn:MMU}\\
&+\sum_{\sigma_R\in\Sigma}z_{\{x_{j,\sigma_R},x_j^{(1,l)}\}}\left(y_{\Phi(I\cup J)\cup\{(x_j,\sigma_R)\}}^*-\sum_{\sigma_L:(\sigma_L,\sigma_R)\in\pi_{i,j}}y_{\Phi(I\cup J)\cup\{(c_i,\sigma_L)\}\cup\{(x_j,\sigma_R)\}}^*\right)_{|I|,|J|\le r}\label{eqn:MRU}\\
&+\sum_{e\in E\setminus(E_{LPath}^{i,l}\cup E_L^{i,l}\cup E_M^{i,j}\cup E_R^{j,l}\cup E_{RPath}^{i,l})}z_e\left(y_{I\cup J\cup\{e\}}'\right)_{|I|,|J|\le r}\label{eqn:MrestU}\\
&-\left(y_{\Phi(I\cup J)}^*-\sum_{(\sigma_L,\sigma_R)\in\pi_{i,j}}y_{\Phi(I\cup J)\cup\{(c_i,\sigma_L)\}\cup\{(x_j,\sigma_R)\}}^*\right)_{|I|,|J|\le r}\label{eqn:MsumU}\\
&+\sum_{(\sigma_L,\sigma_R)\in\pi_{i,j}}\left(\sum_{e\in E_{LPath}^{i,l}\cup E_{RPath}^{i,l}}z_e+z_{\{c_i^{(1,l)},c_{i,\sigma_L}\}}+z_{\{c_{i,\sigma_L}, x_{j,\sigma_R}\}}+z_{\{x_{j,\sigma_R},x_j^{(1,l)}\}}-1\right)\label{eqn:ZsumU}\\
&\quad\times\left(y_{\Phi(I\cup J)\cup\{(c_i,\sigma_L)\}\cup\{(x_j,\sigma_R)\}}^*\right)_{|I|,|J|\le r} \succcurlyeq0 \label{eqn:MLRU}
\end{align}

The matrix in \eqref{eqn:MMU} is positive semidefinite by using lemma \ref{lem:simple_slack} twice.

The matrix in \eqref{eqn:MrestU} is positive semidefinite because it is a principal submatrix of $M_{r+1}(\mathbf{y}')$.

The matrix in \eqref{eqn:MLRU} is positive semidefinite because it equals to $\left(y_{I\cup J\cup\{\{c_i^{(1,l)},c_{i,\sigma_L}\},\{x_{j,\sigma_R},x_j^{(1,l)}\}\}}^*\right)_{|I|,|J|\le r}$, and is also a principal submatrix of $M_{r+1}(\mathbf{y}')$.

\eqref{eqn:ZsumU} is non-negative because $\mathbf{z}$ satisfies that for all path $P\in\mathcal{P}_{c_i^{(k-1,l)},x_j^{(k-1,l)}}$, we have $\sum_{e\in P}z_e\ge1$. Also all $z_e$ are non-negative.

The matrix in \eqref{eqn:MLU}, \eqref{eqn:MRU}, \eqref{eqn:MsumU}, and \eqref{eqn:Mpaths} are all zero matrix for the same reason as \eqref{eqn:ML}, \eqref{eqn:MR}, and \eqref{eqn:Msum}.

Therefore we proved the slack moment matrix $M_r^\mathbf{z}(\mathbf{y}')=\left(\sum_{e\in E}z_ey_{I\cup J\cup\{e\}}'-y_{I\cup J}'\right)_{|I|,|J|\le r}\succcurlyeq0$ where $\mathbf{z}\in\mathcal{Z}^{u,v}$ and $\{u,v\}\in E_{outer}$.

\subsection{Integral Solution for Undirected $(2k-1)$-Spanner}
In this section, we argue that the spanner instance described in previous sections has a \textit{large} optimal solution. More precisely,

\begin{lemma}\label{lem:undirectSound}
The optimal solution of the $(2k-1)$-spanner instance is at least $nkK|\Sigma|\sqrt{K}$.
\end{lemma}
Before proving this Lemma, we first prove the following Claim:
\begin{claim}\label{claim:undirectOuter}
Any feasible solution $S$ of the $(2k-1)$-spanner instance can be transformed to another feasible solution $S'$, where $|S'|\le2|S|$ and $S'\cap E_{Outer}=\varnothing$.
\end{claim} 
\begin{proof}
Every edge in $E_{LPaths}\cup E_{M}\cup E_{RPaths}$ must be in $S$, because they are the only way to span themselves.

For any outer edge $\{u,v\}\in E_{Outer} \cap S$, we can replace $\{u,v\}$ with a path of length $(2k-1)$ from $u$ to $v$ by Claim \ref{claim:path}. We also argue that each time we replace the path, we will introduce at most $2$ new edges into the solution: one in $E_L$ and one in $E_R$.

Therefore, this will at most add a factor of $2$ to the total size, and the new solution does not contain any solution in $E_{Outer}$.
\end{proof}

\paragraph{Proof of Lemma \ref{lem:undirectSound}:}
We argue that the size of the optimal solution must be at least $nK|\Sigma|\sqrt{K}$, otherwise the \textsc{Projection Games} instance has a solution in which $\omega(\frac{m\ln n}{\varepsilon})$ edges are satisfied. This contradicts with Lemma \ref{lem:projInstance}.

If the optimal solution of the $(2k-1)$-spanner instance is less than $nK|\Sigma|\sqrt{K}$, then by Claim \ref{claim:undirectOuter} there is a solution $S'$ with less than $2nK|\Sigma|\sqrt{K}$ edges that does not use any edge from $E_{Outer}$. This implies that there exists $l\in[K|\Sigma|]$, such that $S'\cap E_L^l$ and $S'\cap E_R^l$ both have size less than $2n\sqrt{K}$. This is straightforward from the pigeon hole principle.

Since each vertex in $E_L\cup E_R$ corresponds to a pair in $(L\cup R)\times\Sigma$, we see that $S'\cap(E_L^l\cup E_R^l)$ corresponds to a label assignment $\Psi=\Phi(S'\cap(E_L^l\cup E_R^l))$, where each vertex may be assigned multiple labels in $\Psi$. We also observe that $\Psi$ satisfies all the edges in $E_{Proj}^k$. This is because for each edge $\{c_i,x_j\}\in E_{Proj}$, there is an outer edge $(c_i^l,x_j^l)$, and there is a length $(2k-1)$ path in $S'$ from $c_i^l$ to $x_j^l$. From Claim \ref{claim:path} we know that this related to a label $\sigma_L$ on $c_i$ and a label $\sigma_R$ on $x_j$ which satisfies edge $(c_i,x_j)$.

Next, we will create an assignment $\Psi'$ that satisfies fewer edges, but each vertex in $L$ has at most $\frac{16\sqrt{K}}{n^\varepsilon}$ labels, and each vertex in $R$ has at most $32\sqrt{K}$ labels. We also ignore the vertices in $L$ that have degree larger than $\frac{16\sqrt{K}}{n^\varepsilon}$ labels, there are at most $\frac{2n\sqrt{K}}{\frac{16\sqrt{K}}{n^\varepsilon}}=\frac{n^{1+\varepsilon}}{8}$ of them. We also ignore the vertices in $R$ that have degree larger than $32\sqrt{K}$ labels, there are at most $\frac{2n\sqrt{K}}{32\sqrt{K}}=\frac{n}{16}$ of them. The number of edges that are still satisfied by $\Psi'$ is at least $\frac{n^{1+\varepsilon}K}{2}-\frac{n^{1+\varepsilon}}{8}\cdot K-\frac{n}{16}\cdot 2Kn^\varepsilon=\frac{n^{1+\varepsilon}K}{4}$, because the vertices in $L$ have degree $K$, and the vertices in $R$ have degree at most $2Kn^\varepsilon$ from Lemma \ref{lem:projInstance}.

The final step is to find an assignment that satisfies $\omega(\frac{m\ln n}{\varepsilon})$ edges, where each vertex has only one label. This can be achieved by a probabilistic argument. If we randomly choose a label from existing labels for each vertex in $\Psi'$, then the probability that each edge is still satisfied is $\frac{n^\varepsilon}{16\sqrt{K}}\cdot\frac{1}{32\sqrt{K}}=\frac{n^\varepsilon}{512K}$. Thus in expectation there will be $\frac{n^\varepsilon}{512K}\cdot\frac{n^{1+\varepsilon}K}{2}=\frac{n^{1+2\varepsilon}}{1024}=\omega(\frac{m\ln n}{\varepsilon})$ edges satisfied. This means that there must exist one assignment that satisfies $\omega(\frac{m\ln n}{\varepsilon})$ edges. This contradict with Lemma \ref{lem:projInstance}, and finishes the proof.
\qed

\begin{proof}[\textbf{Proof of Theorem \ref{thm:directSpanner}}]. We can now put our feasible solution analysis and the bound on the integral solution together to get the following integrality gap for direct $(2k-1)$-spanner which finishes the proof: 
${\frac{nK|\Sigma|\sqrt{K}}{mkK|\Sigma|}=\frac{1}{k}\cdot\left(\frac{|V|}{k}\right)^{\min\left\{\frac{1}{18},\frac{5}{32k-6}\right\}-\Theta(\varepsilon)}}$. 
\end{proof} 

\section{Directed Steiner Network and Shallow-Light Steiner Network} \label{app:extension}
In this appendix section we first formally define the \textsc{Directed Steiner Network} problem, and then explain how our integrality gap argument for the \textsc{Directed $(2k-1)$-Spanner} can be modified to prove Theorem \ref{thm:DSN}.

\begin{definition}[\textsc{Directed Steiner Network}]
Given a directed graph $G=(V,E)$, and $p$ pairs of demands $(s_1,t_1),\mathellipsis,(s_p,t_p)$. The objective of \textsc{DSN} is to find a subgraph $G'=(V,S)$ with minimum number of edges, such that for every $i\in[p]$, there is a path from $s_i$ to $t_i$ in $G'$.
\end{definition}

The best known approximation algorithm for this problem also utilizes a flow based LP similar to the spanner LP. Similarly, we can write the $r$-th level of Lasserre hierarchy of \textsc{DSN} as follows: ($\mathcal{Z}^{u,v}=\{\mathbf{z}\in[0,1]^{|E|}\mid\forall P\in\mathcal{P}_{u,v},\sum_{e\in P}z_e\ge1\}$, and $\mathcal{P}_{u,v}$ is the set of all paths from $u$ to $v$):

\[\begin{array}{rll}
&\hspace{-3em}\text{SDP}_{DSN}:&\\
\min&\sum\limits_{e\in E}y_e&\\
s.t.&y_\varnothing=1&\\
&M_{r+1}(y)=\left(y_{I\cup J}\right)_{|I|,|J|\le r+1}\succcurlyeq0&\\
&M_r^\mathbf{z}(y)=\left(\sum_{e\in E}z_ey_{I\cup J\cup\{e\}}-y_{I\cup J}\right)_{|I|,|J|\le r}\succcurlyeq0&\forall i\in[p],\mbox{ and }\mathbf{z}\in\mathcal{Z}^{s_i,t_i}\\
\end{array}\]

The integrality gap instance is almost the same as for the \textsc{Directed $3$-Spanner} problem. There are two differences. The first difference is that we change the number of duplications of $E_{Proj}$ from $2K|\Sigma|$ to $K$. The second difference is that there is no edge set $E_{Outer}$, $E_{LStars}$, and $E_{RStars}$ anymore (and in fact there is no $M_{Paths}$ because $k=2$). Instead, there are demands $(c_i^l,x_j^l)$ for each $l\in[K]$ and $\{c_i,x_j\}\in E_{Proj}$.

Now, the only way from $c_i^l$ to $x_j^l$ is path $\{(c_i^l,c_{i,\sigma_L}),(c_{i,\sigma_L},x_{j,\sigma_R}),(x_{j,\sigma_R},x_j^l)\}$ for some $(\sigma_L,\sigma_R)\in\pi_{i,j}$.

We also slightly modify the solution. Let
\[\Phi(e)=\begin{cases}
\{(c_i,\sigma)\},&\mbox{if }e\in E_L\cup E_M\mbox{ and }e\mbox{ has an endpoint }c_{i,\sigma}\in L_{labels}\\
\{(x_i,\sigma)\},&\mbox{if }e\in E_R\mbox{ and }e\mbox{ has an endpoint }x_{i,\sigma}\in R_{labels}\\
\end{cases}\]
And $\Phi$ can be extended to $\mathscr{P}(E)\rightarrow\mathscr{P}((L\cup R)\times\Sigma)$ by letting $\Phi(S)=\cup_{e\in S}\Phi(e)$. The difference is that edges $e\in E_M$ now have fractional $y_e'$, rather than integral value.

We can prove the moment matrix is positive semidefinite using exactly the same calculations and notations. For the slack moment matrix $M_r^\mathbf{z}(\mathbf{y}')$ where $\mathbf{z}\in\mathcal{Z}^{c_i^l,x_j^l}$, $l\in[K]$ and $\{c_i,x_j\}\in E_{Proj}$  we have:

\begin{align*}
M_r^\mathbf{z}(\mathbf{y}')=&\left(\sum_{e\in E}z_ey_{I\cup J\cup\{e\}}'-y_{I\cup J}'\right)_{|I|,|J|\le r}\\
=&\sum_{e\in E}z_e\left(y_{\Phi(I\cup J)\cup\Phi(e)}^*\right)_{|I|,|J|\le r}-\left(y_{\Phi(I\cup J)}^*\right)_{|I|,|J|\le r}\\
=&\sum_{\sigma_L\in\Sigma}z_{(c_i^l,c_{i,\sigma_L})}\left(y_{\Phi(I\cup J)\cup\{(c_i\sigma_L)\}}^*\right)_{|I|,|J|\le r}\\
&+\sum_{(\sigma_L,\sigma_R)\in\pi_{i,j}}z_{(c_{i,\sigma_L}, x_{j,\sigma_R})}\left(y_{\Phi(I\cup J)\cup\{(c_i,\sigma_L)\}}^*\right)_{|I|,|J|\le r}\\
&+\sum_{\sigma_R\in\Sigma}z_{(x_{j,\sigma_R},x_j^l)}\left(y_{\Phi(I\cup J)\cup\{(x_j,\sigma_R)\}}^*\right)_{|I|,|J|\le r}\\
&+\sum_{e\in E\setminus(E_L^{i,l}\cup E_M^{i,j}\cup E_R^{j,l})}z_e\left(y_{I\cup J\cup\{e\}}'\right)_{|I|,|J|\le r}\\
&-\left(y_{\Phi(I\cup J)}^*\right)_{|I|,|J|\le r}\\
=&\sum_{\sigma_L\in\Sigma}z_{(c_i^l,c_{i,\sigma_L})}\left(y_{\Phi(I\cup J)\cup\{(c_i\sigma_L)\}}^*-\sum_{\sigma_R:(\sigma_L,\sigma_R)\in\pi_{i,j}}y_{\Phi(I\cup J)\cup\{(c_i,\sigma_L)\}\cup\{(x_j,\sigma_R)\}}^*\right)_{|I|,|J|\le r}\\
&+\sum_{(\sigma_L,\sigma_R)\in\pi_{i,j}}z_{(c_{i,\sigma_L}, x_{j,\sigma_R})}\left(y_{\Phi(I\cup J)\cup\{(c_i,\sigma_L)\}}^*-y_{\Phi(I\cup J)\cup\{(c_i,\sigma_L)\}\cup\{(x_j,\sigma_R)\}}^*\right)_{|I|,|J|\le r}\\
&+\sum_{\sigma_R\in\Sigma}z_{(x_{j,\sigma_R},x_j^l)}\left(y_{\Phi(I\cup J)\cup\{(x_j,\sigma_R)\}}^*-\sum_{\sigma_L:(\sigma_L,\sigma_R)\in\pi_{i,j}}y_{\Phi(I\cup J)\cup\{(c_i,\sigma_L)\}\cup\{(x_j,\sigma_R)\}}^*\right)_{|I|,|J|\le r}\\
&+\sum_{e\in E\setminus(E_L^{i,l}\cup E_M^{i,j}\cup E_R^{j,l})}z_e\left(y_{I\cup J\cup\{e\}}'\right)_{|I|,|J|\le r}\\
&-\left(y_{\Phi(I\cup J)}^*-\sum_{(\sigma_L,\sigma_R)\in\pi_{i,j}}y_{\Phi(I\cup J)\cup\{(c_i,\sigma_L)\}\cup\{(x_j,\sigma_R)\}}^*\right)_{|I|,|J|\le r}\\
&+\sum_{(\sigma_L,\sigma_R)\in\pi_{i,j}}\left(z_{(c_i^l,c_{i,\sigma_L})}+z_{(c_{i,\sigma_L}, x_{j,\sigma_R})}+z_{(x_{j,\sigma_R},x_j^l)}-1\right)\\
&\quad\times\left(y_{\Phi(I\cup J)\cup\{(c_i,\sigma_L)\}\cup\{(x_j,\sigma_R)\}}^*\right)_{|I|,|J|\le r} \succcurlyeq0
\end{align*}

All matrices in the above sum are positive semidefinite for the same reasons explained in Section \ref{sec:directFeasibility}.

For the objective value, we know from definition of $\text{SDP}_{Proj}^{r+2}$ and Claim \ref{claim:pair=0} that
\[y_{(c_i,\sigma_L)}^*=\sum_{\sigma_R\in\Sigma}y_{(c_i,\sigma_L),(x_j,\sigma_R)}^*=0+y_{(c_i,\sigma_L),(x_j,\pi_{i,j}^{-1}(\sigma_L))}^*\]
Therefore, the objective value is
\begin{align*}
\sum_{e\in E_L}y_e'+\sum_{e\in E_M}y_e'+\sum_{e\in E_R}y_e'
=&K\cdot\sum_{i\in|L|}\sum_{\sigma\in\Sigma}y_{(c_i,\sigma)}^*+\sum_{\{c_i,x_j\}\in E_{Proj}}\sum_{(\sigma_L,\sigma_R)\in\pi_{i,j}}y_{(x_j,\sigma_R)}^*+K\cdot\sum_{i\in|R|}\sum_{\sigma\in\Sigma} y_{(x_i,\sigma)}^*\\
=&K|L|+\sum_{\{c_i,x_j\}\in E_{Proj}}\sum_{(\sigma_L,\sigma_R)\in\pi_{i,j}}y_{(c_i,\sigma_L),(x_j,\sigma_R)}^*+K|R|\\
=&K|L|+K|L|\cdot1+K|R|=O(n^{1+\varepsilon}K)
\end{align*}

The optimal solution is at least $nK\sqrt{K}$ based on the same analysis as in the \textsc{Directed $3$-Spanner} problem. Thus we have an $\frac{nK\sqrt{K}}{n^{1+\varepsilon}K}=|V|^{\frac{1}{16}-\Theta(\varepsilon)}$ integrality gap, which concludes the proof of Theorem \ref{thm:DSN}.

\begin{definition}[\textsc{Shallow-Light Steiner Network}]
Given a graph $G=(V,E)$, a distance bound $L$, and $p$ pairs of vertices $\{s_1,t_1\},\mathellipsis,\{s_p,t_p\}$. The objective of \textsc{SLSN} is to find a subgraph $G'=(V,S)$ with minimum number of edges, such that for every $i\in[p]$, there is a path between $s_i$ and $t_i$ in $G'$ with length less or equal to $L$.
\end{definition}

The LP and SDP for \textsc{Shallow-Light Steiner Network} is the same as $\text{SDP}_{DSN}$, except that the $\mathcal{P}_{u,v}$ becomes the set of paths between $u$ and $v$ within distance $L=3$.

The instance remains the same, with the only difference that the edges are now undirected. The only way to connect $c_i^l$ and $x_j^l$ within distance $3$ is a path $\{\{c_i^l,c_{i,\sigma_L}\},\{c_{i,\sigma_L},x_{j,\sigma_R}\},\{x_{j,\sigma_R},x_j^l\}\}$ for some $(\sigma_L,\sigma_R)\in\pi_{i,j}$. The solution, objective value, and optimal solution are exactly the same as Direct Steiner Network. Thus we have the same integrality gap, proven Theorem \ref{thm:SLSN}.

\end{document}